\documentclass[reqno,11pt]{article}
\usepackage{mystyle}
\usepackage{authblk}
\usepackage{amsmath,amssymb,epsfig,amsthm,bm,xcolor}
\usepackage{longtable,tabu}
\newcommand{\Argmax}{\mathop{\mathrm{argmax}}}

\usepackage{appendix}

\newcommand{ \ep }{\varepsilon}

\newcommand{\Pro}{\mathcal{P}}

\newcommand{\suchthat}{\ensuremath{\ : \ }} 

\newcommand{\argmax}{\text{arg}~\text{max}}

\newcommand{\restr}[1]{\lower3pt\hbox{$|_{#1}$}}

\newcommand{\argmmax}{\mathop{\mathrm{argmax}}}

\newcommand{\Tr}{\operatorname{Tr}}

\newcommand{\DM}{\mathfrak{P}}

\newcommand{\unit}{1\!\!1}
\newcommand{\prim}{\mathfrak{F}^{\ep}}

\newcommand{\primal}{\tF^{\ep}}
\newcommand{\primin}{\mathfrak{F}^{\ep}}

\newcommand{\dual}{\mathfrak{D}^{\ep}}
\newcommand{\dualf}{\tD_{\bm \gamma}^\ep}
\newcommand{\dualin}{\mathfrak{D}^{\ep}}

\newcommand{\Sink}{\mathcal{T}^\ep}

\newcommand{\ffD}{\text D_\gamma^{-,\ep}}
\newcommand{\ffd}{\text D_\gamma^-}
\newcommand{\fbD}{\text D_\gamma^{+, \ep}}

\newcommand{\fbfD}{\text D_\gamma^{\pm, \ep}}

\newcommand{\fD}{\mathfrak D_\ep^-}
\newcommand{\bD}{\mathfrak D_\ep^+}
\newcommand{\bfD}{\mathfrak D_\pm^\ep}

\theoremstyle{plain}
\newtheorem{theo}{Theorem}[section]
\newtheorem{proposition}{Proposition}[section]

\theoremstyle{definition}
\newtheorem{defin}[theo]{Definition}
\newtheorem{rem}[theo]{Remark}

\newcommand{\cM}{\mathcal{M}}

\newcommand{\cS}{\mathcal{S}}

\newcommand{\bR}{\mathbb{R}}

\newcommand{\bN}{\mathbb{N}}
\newcommand{\bC}{\mathbb{C}}

\newcommand{\tP}{\operatorname{P}}
\newcommand{\tQ}{\operatorname{Q}}
\newcommand{\tT}{\operatorname{T}}
\newcommand{\tD}{\operatorname{D}}
\newcommand{\tF}{\operatorname{F}}
\newcommand{\tS}{\operatorname{S}}
\newcommand{\ren}{\operatorname{Ren}}

\newcommand{\fT}{\mathfrak{T}}

\DeclareFontFamily{U}{mathx}{\hyphenchar\font45}
\DeclareFontShape{U}{mathx}{m}{n}{
	<5> <6> <7> <8> <9> <10>
	<10.95> <12> <14.4> <17.28> <20.74> <24.88>
	mathx10
}{}
\DeclareSymbolFont{mathx}{U}{mathx}{m}{n}
\DeclareMathSymbol{\bigtimes}{1}{mathx}{"91}

\newcommand{\dd}{\, \mathrm{d}}

\usepackage{scalerel}

\usepackage{color}
\usepackage{pstricks}
\usepackage{ifthen}
\newrgbcolor{daricolor}{.7 .1 .1}
\newrgbcolor{lorecolor}{.1 .7 .1}
\newrgbcolor{augucolor}{.1 .1 .7}
\usepackage{csquotes}

\newboolean{DEBUG}
\setboolean{DEBUG}{true}

\ifthenelse {\boolean{DEBUG}}
{}

\ifthenelse {\boolean{DEBUG}}
{}

\ifthenelse {\boolean{DEBUG}}
{}

\ifthenelse {\boolean{DEBUG}}
{}

\title{A Non-Commutative Entropic Optimal Transport Approach to Quantum Composite Systems at Positive Temperature}

\author[1]{Dario Feliciangeli}
\author[2]{Augusto Gerolin}
\author[1]{Lorenzo Portinale}
\affil[1]{IST Austria, Am Campus 1, 3400 Klosterneuburg, Austria.}
\affil[2]{Department of Theoretical Chemistry, Vrije Universiteit Amsterdam, The Netherlands.}

\date{\textit{Dedicated to Tito.}}
\begin{document}
\maketitle
\begin{abstract}
	This paper establishes new connections between many-body quantum systems, One-body Reduced Density Matrices Functional Theory (1RDMFT) and Optimal Transport (OT), by interpreting the problem of computing the ground-state energy of a finite dimensional composite quantum system at positive temperature as a non-commutative entropy regularized Optimal Transport problem. 
	We develop a new approach to fully characterize the dual-primal solutions in such non-commutative setting.	
	The mathematical formalism is particularly relevant in quantum chemistry: numerical realizations of the many-electron ground state energy can be computed via a non-commutative version of Sinkhorn algorithm. 
	Our approach allows to prove convergence and robustness of this algorithm, which, to our best knowledge, were unknown even in the two marginal case.
	 Our methods are based on careful a priori estimates in the dual problem, which we believe to be of independent interest.
	Finally, the above results are extended in 1RDMFT setting, where bosonic or fermionic symmetry conditions are enforced on the problem.
\end{abstract}
\tableofcontents

\section{Introduction}

In this work we are interested in studying the ground state energy of a finite dimensional composite quantum system at positive temperature. In particular, we focus on the problem of minimizing the energy of the composite system \textit{conditionally} to the knowledge of the states of all its subsystems. 

The first motivation for this study is physical: it is useful to understand how one could infer the state of a composite system when one only has experimental access to the measurement of the states of its subsystems. The second motivation is mathematical: indeed this problem can be cast as a non-commutative optimal transport problem, therefore showcasing how several ideas and concepts introduced in the commutative setting carry through to the non-commutative framework. Finally, a third motivation comes from the fact that one-body reduced density matrix functional theory, which is of interest on its own, can be framed as a special case of our setting.

Let us consider a composite system with $N$ subsystems, each with state space given by the complex Hilbert space $\mathfrak{h}_j$  of dimension $d_j<\infty$, for $j=1,\dots, N$, and denote the state space of the composite system $\mathfrak{h}:=\mathfrak{h}_1 \otimes \mathfrak{h}_2 \otimes \dots \otimes \mathfrak{h}_N$ (with dimension $d=d_1 \cdot d_2 \cdot \dots d_N$). Further denote by $\H$ the Hamiltonian to which the whole system is subject and suppose that $\H=\H_0+\H_{\text{int}}$, where $\H_0$ is the non-interacting part of the Hamiltonian, i.e.  $\H_0=\bigoplus_{j=1}^N \H_j:=\H_1 \otimes \unit \dots \otimes \unit+ \unit\otimes \H_2 \otimes \unit \dots \otimes \unit +\dots +\unit\otimes \dots \otimes \unit \otimes \H_N$ with $\H_j$ acting on $\mathfrak{h}_j$, and $\H_{\text{int}}$ is its interacting part. Finally, suppose to have knowledge of the states $\bm \gamma=(\gamma_1, \dots, \gamma_N)$ of the $N$ subsystems, where each $\gamma_j$ is a density matrix over $\mathfrak{h}_j$. 

Then the energy of the composite system at temperature $\varepsilon > 0$ is given by
\begin{align}
	\label{eq:multimarginalGSEposT}
	\inf_{\Gamma \mapsto \bm \gamma} \left\{\Tr(\H\Gamma)+\varepsilon S(\Gamma)\right\}&= \sum_{j=1}^N\Tr(\H_j \gamma_j)+\primin(\bm \gamma)\nonumber\\
	&:=\sum_{j=1}^N\Tr(\H_j \gamma_j)+\inf_{\Gamma \mapsto \bm \gamma} \left\{\Tr(\H_{\text{int}}\Gamma)+\varepsilon S(\Gamma)\right\},
\end{align}
where the shorthand notation $\Gamma \mapsto \bm \gamma$ denotes the set of density matrices over $\mathfrak{h}$ with $j$-th marginal equal to $\gamma_j$, and $S(\Gamma):=\Tr\left(\Gamma \log(\Gamma)\right)$ is the opposite of the Von Neumann entropy of $\Gamma$ (note that we prefer to adopt the mathematical sign convention).

Our approach for the study of $\primin(\bm \gamma)$ borrows ideas from optimal transport and convex analysis, and takes the following observation as a starting point: the minimization appearing in $\primin$ can be cast as a non-commutative entropic optimal transport problem. Indeed, one looks for an optimal non-commutative coupling $\Gamma$, with fixed non-commutative marginals (i.e. partial traces) $\bm \gamma$, which minimizes the sum of a transport cost (given by $\Tr(\H_{\text{int}}\Gamma)$) and an entropic term. In light of this interpretation, setting the quantum problem at positive temperature $\varepsilon$ corresponds to consider an entropic optimal transport problem with parameter $\varepsilon$. 

Guided by this viewpoint, we first show that $\primin$ has a dual formulation (see Theorem \ref{theo:duality} (i)), i.e. that the constrained minimization appearing in its definition is in duality with an unconstrained maximization problem (defined in \eqref{eq:dual_statements}). We can then consider any vector $(U^{\varepsilon}_1,\dots,U^{\varepsilon}_N)$ of self-adjoint matrices which is a maximizer in the dual functional of $\primin$, whose existence and uniqueness up to trivial transformations we prove in Theorem \ref{theo:duality}(ii). We refer to such $U^{\ep}_i$-s as \textit{Kantorovich potentials} and show in Theorem \ref{theo:duality}(iii) that the unique minimizer $\Gamma^{\varepsilon}$ realizing $\primin(\bm\gamma)$ can be written in terms of them as
\begin{align}
	\label{eq:shape}
	\Gamma^{\varepsilon}=\exp\left( \frac {\bigoplus_{i=1}^N  U^{\varepsilon}_i-\H_{\text{int}}}{\ep}\right),
\end{align} 
in the case of all the $\gamma_j$-s having trivial kernels (in the general case a very similar formula holds). In this setting, $\primin$ is continuous and its functional derivative can be computed in terms of the \textit{Kantorovich potentials} as
\begin{align}
	\label{eq:FunctDer}
	\frac{\dd\primin}{\dd \gamma_i} (\bm\gamma)=U_i^{\varepsilon}, \quad \text{for all} \quad i=1,\dots, N,
\end{align}
as we show in Proposition \ref{prop:OTcont}.

Furthermore, we introduce the Non-Commutative Sinkhorn algorithm to compute the optimizer realizing $\primin(\bm \gamma)$. This algorithm exploits the shape of the minimizer obtained in \eqref{eq:shape}, in order to construct a sequence $\Gamma^{(k)}$ of density matrices converging to $\Gamma^{\ep}$ of the form 
\begin{align}
	\Gamma^{(k)}=\exp\left(\frac{\bigoplus_{i=1}^N U_i^{(k)}-\H_{\text{int}}}{\varepsilon}\right),
\end{align}
where the vector $(U_1^{(k)},\dots,U_N^{(k)})$ is iteratively updated by progressively imposing that $\Gamma^{(k)}$ has at least one correct marginal. We prove the convergence and the robustness of this algorithm in Section \ref{sec:noncomsin}. 
%

It is important to note that studying $\primin(\bm \gamma)$, i.e. the constrained minimization at fixed marginals, can also help solving the unconstrained minimization of the Hamiltonian $\H$ at positive temperature $\ep$. Indeed, denoting by $\DM(\mathfrak{h})$ the set of density matrices over $\mathfrak{h}$, then
\begin{align}
	\label{eq:Uncostrained}
	E^{\ep}(\H):=\inf_{\Gamma \in \DM(\mathfrak{h})} \left\{\Tr(\H\Gamma)+\ep S(\Gamma)\right\}=\inf_{\bm \gamma} \left\{\sum_{j=1}^N\Tr(\H_j \gamma_j)+\primin(\bm \gamma)\right\}.
\end{align}
Combining \eqref{eq:FunctDer} and \eqref{eq:Uncostrained} allows to write down the Euler--Lagrange equation of \eqref{eq:Uncostrained} recovering its optimizer, i.e. the Gibbs state constructed with $\H$ at temperature $\ep$.

Our work is not the first to try to extend the theory of optimal transport to the non-commutative setting. One of the first attempts was carried out by E. Carlen and J. Maas \cite{CarMas14}, followed by many others (e.g. \cite{BreVor20, CalGolPau18, CalGolPau19, CheGanGeoTan19, CheGeoan18, DPaTre19, DPaTreGioAmb18, GeoPav15, MitMie17, MonVor20, PeyChiViaSol19}). There is an important distinction to be made here. Commutative optimal transport can be cast \textit{equivalently} as a static coupling problem or as a dynamical optimization problem. On the other hand, in the non-commutative setting it is not clear what is the relation (if any) between the two interpretations. This singles out a big difference between works that consider the dynamical formulation of commutative optimal transport as a starting point (e.g. \cite{BreVor20, CarMas14, CheGanGeoTan19, CheGeoan18, MitMie17, MonVor20, PeyChiViaSol19}) and the ones which instead focus on its static formulation (e.g. \cite{Cut13,GalSal,LeoSurvey,Schr31,Zam86}). 
%
%

This paper adopts an even different approach. We consider as a starting point the Entropic regularization of optimal transport (which is to be considered as an extension of static optimal transport, see e.g. the survey \cite{LeoSurvey} and references therein)
and introduce its non-commutative counterpart. We carry out this program by extending the method developed in \cite{DMaGer19,DMaGer20,GerKauRajEnt}. See also Section \ref{sec:noncomsin} for a detailed explanation of the multimarginal Sinkhorn algorithm in the commutative setting, as studied in \cite{DMaGer19}.

In the work \cite{CalGolPau18}, the authors study the case of $\eps=0$ temperature and prove a duality result for the non-commutative problem in the very same spirit of the Kantorovich duality for the classical Monge problem. The recent work \cite{Wirth:2021} studies the entropic quantum optimal transport problem as well, adopting, in constrast to our static approach, a dynamical formulation. Therein, the author proves a dynamical duality result at positive and zero temperature. 
To the best of our knowledge, the present work is the first complete analysis of the quantum entropic transport problem in the static framework.

As for the Sinkhorn Algorithm, another concept which we borrow from the commutative setting and extend to the quantum one, its convergence in the commutative setting was first established in the $N=2$ marginal case \cite{FraLor89,Sin64} for discrete measures and in \cite{RusIPFP} for continuous measures (see also \cite{chen2016entropic}). In the multi-marginal setting, convergence guarantees were obtained for the discrete case in \cite{ChiPeySchVia18,KarlRin17} and for continuous measures in \cite{DMaGer19,DMaGer20}. Other variants of the Sinkhorn algorithm for (unbalanced) tensor-valued measures or matrix optimal mass transport have been studied in \cite{PeyChiViaSol19, RyuCheWucOsh18} and do not apply to our setting. In the context of Computational Optimal Transport, the entropic regularization and the Sinkhorn algorithm was introduced in \cite{Cut13, GalSal}.

\subsection*{Enforcing symmetry constraints: One-body Reduced Density Matrix Functional Theory}

We conclude this introduction by briefly discussing the case in which symmetry conditions are enforced on the problem, either bosonic or fermionic, which we can also treat (see Section \ref{sec:symmetryresults}). In this case, \eqref{eq:multimarginalGSEposT} makes sense only for $\mathfrak{h}_j=\mathfrak{h}_0$ for all $j=1,\dots,N$ and $\bm \gamma=(\gamma,\dots, \gamma)$ (i.e. the underlying Hilbert spaces and the marginals must all be the same) and its study can be framed in the context of One-body Reduced Density Matrix Functional Theory ($1$RDMFT), introduced in 1975 by Gilbert \cite{Gil75} as an extension of the Hohenberg-Kohn (Levy-Lieb) formulation of Density Functional Theory (DFT) \cite{HohKoh,Levy,Lieb83}. In the last decades, DFT and $1$RDMFT have been standard methods for numerical electronic structure calculations and are to be considered a major breakthrough in fields ranging from materials science to chemistry and biochemistry. \\
In both these theories one tries to approximate a complicated N-particle quantum system by studying one-particle objects, namely one-body densities in the case of DFT and one-body reduced density matrices in the case of 1RDMFT, by using a two-steps minimization analogous to the one introduced in \eqref{eq:Uncostrained}. 

It is interesting to see that the well-known Pauli principle (see e.g. \cite[Theorem 3.2]{LiebSeiringer10}) 
, which provides necessary and sufficient conditions for $\gamma$ to be the one-body reduced density matrix of an $N$-body antisymmetric density matrix, finds a variational interpretation in our discussion. Indeed, in the antisymmetric case we show (see Proposition \ref{prop:pauli}) that $\gamma$ satisfies the Pauli principle (resp. satisfies the Pauli principle \textit{strictly}) if and only if the supremum of the dual functional of $\primin$ is finite (resp. is attained), as it is to be expected. 

Other extensions of DFT have been considered, including Mermin's Thermal Density Functional Theory \cite{Mer65}, Spin DFT \cite{vonBarHed72}, and Current DFT \cite{VigRas87}. Physical and computational aspects of 1RDMFT have been investigated in \cite{ayers2005generalized, BalEicGro15,BalEicGro12, BalEicGro12c, BuiBae02, Men15, Mul84,Per05,PerGie15, ReqPan08,Sch19,Lee07}. A framework for {1RDMFT} for Bosons at zero temperature was recently introduced in \cite{BenWolMigSch20} (see also \cite{GieRug19} and references therein for a recent review). In particular, the first exchange-correlation energy in density-matrix functional theory was introduced by M\"uller \cite{Mul84}, leading to mathematical results \cite{FraLieSei07,FraNamVan18}.

\subsection*{Organisation of the paper}

The paper is divided as follows: in Section \ref{sec:contributions} we introduce the framework, the main definitions, and present our main results Theorem \ref{theo:duality}, Theorem \ref{theo:introsinkhorn}, and Theorem \ref{theo:dualBosFerm}. In Section \ref{sec:prem} we introduce and develop the technical tools needed to prove our main results, in particular we define the notion of non-commutative $(\H,\ep)$-transform (see Section \ref{sec:Heptransform}) and prove a stability and differentiability result for the primal problem in Proposition \ref{prop:OTcont}. In Section \ref{sec:mainresults}, Section \ref{sec:noncomsin}, and Section \ref{sec:1RDM} we build upon Section \ref{sec:prem} and prove our main results, respectively, Theorem \ref{theo:duality}, Theorem \ref{theo:introsinkhorn}, and Theorem \ref{theo:dualBosFerm}.

\section{Contributions and statements of the main results}\label{sec:contributions}
The main contributions of this work consist in 
\begin{itemize}
	\item Theorem \ref{theo:duality}, which represents a duality result for the functional $\primin$ (whose definition is recalled below in equation \eqref{eq:primal_statements}). Theorem \ref{theo:duality} also includes the characterization of the optimizers of $\primin$ (and of its dual functional). 
	\item The introduction of a non-commutative Sinkhorn algorithm, which can be used to compute the aforementioned optimizers. We also prove convergence and robustness of this algorithm in Theorem \ref{theo:introsinkhorn}.
	\item The generalization of Theorem \ref{theo:duality} to the case of bosonic or fermionic systems, stated in Theorem \ref{theo:dualBosFerm}. This also allows to give an interesting variational characterization of the Pauli exclusion principle (see Proposition \ref{prop:pauli}).
\end{itemize}

In our analysis, the main tool is the non-commutative $(\H, \eps)$-transform (introduced in Section \ref{sec:Heptransform}), which allows to obtain a priori estimates on
maximizing sequences of Kantorovich potentials, yielding compactness. Although this approach is not strictly necessary to prove duality (Theorem \ref{theo:duality})
in our finite dimensional setting, we believe these estimates to be of independent interests.
Moreover, they are fundamental to prove the convergence of the non-commutative
Sinkhorn algorithm (Theorem \ref{theo:introsinkhorn}).

We now proceed to introduce our setting and state our main contributions.

\subsection{Duality and minimization of $\primin$}

We recall that in this case we simply work with a general composite system, with no symmetry constraints enforced. For $d \in \N$, we shall denote by $\cM^d = \cM^d(\bC)$ the set of all $d \times d$ complex matrices, by $\cS^d$ the hermitian elements of $\cM^d$, and by $\cS_\geq^d$ (respectively $\cS_>^d$) the set of all the positive semidefinite (resp. positive definite) elements of $\cS^d$. With a slight abuse of notation, we denote by $\Tr$ the trace operator on $\cM_d$ for any dimension $d$.  Furthemore, for any Hilbert space $\mathfrak{h}$, we denote by $\DM(\mathfrak{h})$ the set of \textit{density matrices} over $\mathfrak{h}$, namely the positive self-adjoint operators with trace one. For simplicity, we shall also use the notation $\DM^d=\DM(\bC^d)$. For every $N \in \N$ we adopt the notation $[N]:= \{1, ... , N\}$.

Our main object of study is the minimisation problem for $N \in \N$, $i \in [N]$, $\gamma_i \in \DM^{d_i}$, $\H \in \cS^{\bm d}$ 
\begin{align}	\label{eq:primal_statements}
	\prim(\bm \gamma) = \inf \left\lbrace \Tr(\H\Gamma)+ \eps  \Tr(\Gamma \log \Gamma) \suchthat \Gamma\in\DM^{\bm d} \text{ and } \Gamma\mapsto \bm \gamma  \right\rbrace,
\end{align}
where $d_i \in \N$, $\bm d: =\prod_{i=1}^N d_i$,  $\bm \gamma := (\gamma_i)_{i \in [N]}$, and $\Gamma\mapsto \bm \gamma $ means that the $i$-th marginal \eqref{eq:defP_i} of $\Gamma$ is equal to $\gamma_i$. This coincides with the Definition of $\primin$ given in \eqref{eq:multimarginalGSEposT}. 

The natural space to work with is given by
$	\mathcal O:= \bigotimes_{i=1}^N \big( \ker \gamma_i \big)^\perp 
$
%
where for simplicity we set $\hat d_i := (d_i - \dim \ker \gamma_i)$ and $\mathbf{\hat{d}}: =\prod_{i=1}^N \hat d_i$. We also denote by $\H_{\mathcal O}$ the restriction of $\H$ to the subspace $\mathcal O$.  The corresponding dual problem is defined as
\begin{gather}	\label{eq:dual_statements}
	\dualin(\bm \gamma) = \sup \Bigg\{  \sum_{i=1}^N \Tr(U_i \gamma_i)-\eps \Tr \bigg( \exp\bigg[ \frac{\bigoplus_{i=1}^N U_i - \H_{\mathcal O}}{\eps} \bigg] \bigg) \suchthat U_i\in \cS^{\hat  d_i} \Bigg\} + \eps  \, ,
\end{gather}
where $\bigoplus$ denotes the Kronecker sum \eqref{eq:def_Kronecker}.

Our first result is a duality result and serves also as a characterization of the minimizers in \eqref{eq:primal_statements}.
Note that, throughout the whole paper, when no confusion can arise, we shall use the slightly imprecise notation $\alpha \unit = \alpha$ for $\alpha \in \bC$.

\begin{theo}[Duality]
	\label{theo:duality}
	Let $\ep >0$, $N \in \N$, and $\H \in \cS^{\bm d}$. For fixed $\bm \gamma =( \gamma_i \in \DM^{d_i} )_{i \in [N]} $, consider the primal and dual problems $\primin(\bm \gamma)$, $\dualin(\bm \gamma)$ as in \eqref{eq:primal_statements}, \eqref{eq:dual_statements} respectively. We then have that
	\begin{itemize}
		\item[(i)] the primal and dual problems coincide, i.e. $\primin (\bm \gamma) = \dualin(\bm \gamma)$.
		\item[(ii)] $\dualin(\bm \gamma)$ admits a maximizer $\{   U_i^\ep \in \cS^{\hat d_i} \}_{i=1}^N$, which is unique up to trival translation. Precisely, if $\{\tilde U_i^\ep \in \cS^{\hat d_i} \}_{i=1}^N$ is another maximizer, then $\tilde U_i^\ep - U_i^\ep = \alpha_i \in \R$ with $\sum_i \alpha_i =0$.
		\item[(iii)] There exists a unique $ {\Gamma}^{\ep}\in \DM^{\bm d}$ with $ {\Gamma}^{\ep} \mapsto \bm \gamma$ which minimizes the functional $\primin(\bm \gamma)$. Moreover, $ {\Gamma}^{\ep}$ and $\{   U_i^\ep \}$ are related via the formula
		\begin{align}	\label{eq:optimality_maintheorem}
			{\Gamma}^{\ep}=\exp\left( \frac {\bigoplus_{i=1}^N   U_i^\ep-\H_{\mathcal O}}{\ep}\right) \quad \text{on} \; \mathcal O
		\end{align}
		and $\Gamma^{\ep}=0$ on $\mathcal O^{\perp}$.
		
	\end{itemize}  
\end{theo}
The proof of the existence of maximizers for the dual problem follows the direct method of Calculus of Variations. In analogy with \cite{DMaGer19,DMaGer20}, where the notion of commutative $(c,\eps)$-trasform is introduced, we define the non-commutative $(\H,\ep)$-transform (see Section \ref{sec:Heptransform}). We use this tool to obtain a priori estimates on $U$ and infer compactness of the maximizing sequences of Kantorovich potentials.

%
As a byproduct of the a priori estimates obtained in Section \ref{sec:Heptransform}, it is possible to prove a stability result (with respect to the marginals) for the Kantorovich potentials and compute the Frech\'et derivative of $\prim(\cdot)$. This is the content of the following proposition, which is proved in Section \ref{sec:mainresults}.
For simplicity, we here assume that the marginals have trivial kernel. With a bit more effort, and arguing as in Theorem \ref{theo:duality} (see also Remark \ref{rem:kernels}), one can obtain a similar result in the general setting as well.

\begin{proposition}[Stability and differentiability of $\prim(\cdot)$]	\label{prop:OTcont}
	Fix $\eps>0$ and assume $\ker(\gamma_i) = \{ 0 \}$.
	\begin{enumerate}
		\item[$(i)$] \emph{Stability}: if $\bm \gamma^n = (\gamma^n_i)_{n\in\N} $, $\gamma^n_i \subset \DM^{d_i}$ is a sequence of density matrices converging to $\bm \gamma = (\gamma_i)_{n\in\N}$ as $n \to \infty$, then any sequence of Kantorovich potentials $\bm U^{\ep,n}$ converges, up to subsequences and renormalisation, to a Kantorovich potential $\bm U^\ep$ for $\prim(\bm \gamma)$. Therefore, the functional $\prim(\cdot)$ is continuous.
		\item[$(ii)$] \emph{Frech\'et differential}:  $\prim(\cdot)$ is Fr\'echet differentiable and for every $i \in [N]$ it holds
		\begin{align}	\label{eq:FrechDiff}
			\Big(\dfrac{d\prim}{d\gamma_i}\Big)_{\bm \gamma}(\sigma) = \Tr \big( \bm U^\ep_i \sigma \big), \quad \forall \sigma \in \cS^{d_i}, \; \Tr(\sigma) =0 \, , 
		\end{align}
		where $\bm U^\ep$ is a Kantorovich potential for $\prim(\bm \gamma)$.
	\end{enumerate}
\end{proposition}

As derived in \cite{Gil75} and explained, for instance, in \cite{Per05}, the relevance of the functional derivative in the 1RDMFT case is to find an eigenvalue equation to find an efficient optimization for the one-particle eigenvalue equations.

\subsection{Non-commutative Sinkhorn algorithm}

The second contribution of this work is to introduce and  prove the convergence of a non-commutative Sinkhorn algorithm (see Section \ref{sec:noncomsin}), aimed at computing numerically the optimal density matrix $\Gamma^{\ep}$ and the corresponding Kantorovich potentials $\{ U_i^{\ep} \}_i$.

For this purpose, we define non-commutative  $(\H,\eps)$-transform operators, which extend the notion of $(c,\ep)-$transforms as introduced in \cite{DMaGer19} (see also Section \ref{sec:noncomsin} for a detailed explanation). Note that the $(\H,\eps)$-transform also depends on $\bm \gamma$, but we omit this dependence as $\bm \gamma$ is a fixed parameter of the problem. 

For $i \in [N]$ and $\eps>0$, we consider the operators $\Sink_i : \bigtimes_{j=1}^N \cS^{\hat d_j} \to \bigtimes_{j=1}^N \cS^{\hat d_j}$ of the form 
\begin{align*}
	\bm U:= (U_1, \dots, U_N), \quad \big(\Sink_i(\bm U) \big)_j = 
	\begin{cases}
		U_j &\text{if } j \neq i, \\
		\fT_i^\ep(U_1, \dots, U_{i-1}, U_{i+1} , \dots U_N) &\text{if }j=i
	\end{cases} 
\end{align*} 
where $\fT_i^\ep$ is defined implicitly via 
\begin{align}	\label{eq:def_implicit-T_i_intro}
	\tP_i \left[ \exp\left(\frac {\bigoplus_{j=1}^N \big( \Sink_i(\bm U) \big)_j-\H_{\mathcal O}}{\ep}\right) \right] = \gamma_i
\end{align}
and $\tP_i$ denotes the $i$-th marginal operator, obtained by tracing out all but the $i$-th coordinate, see \eqref{eq:defP_i}. In Section \ref{sec:noncomsin}, we show that the maps $\Sink_i$ are well-defined, i.e. the equation \eqref{eq:def_implicit-T_i_intro} admits a unique solution $\Sink_i(\bm U)$. 

Note that, by construction, the matrix $\exp\left(\bigoplus_{i=1}^N \left(\left(\Sink_i(\bm U) \right)_j-\H_{\mathcal O}\right)/\ep\right) \in \DM^{\mathbf{\hat d}}$ and it has the $i$-th marginal equal to $\gamma_i$. The non-commutative Sinkhorn algorithm is then defined by iterating this procedure for every $i\in[N]$. We define the one-step Sinkhorn map as 
\begin{align*}
	\begin{gathered}
		\tau: \bigtimes_{j=1}^N \cS^{\hat d_j} \to \bigtimes_{j=1}^N \cS^{\hat d_j} , \\
		\tau(\bm U) := (\Sink_N \circ \dots \circ \Sink_1) (\bm U).
	\end{gathered}
\end{align*}

The Sinkhorn algorithm is obtained by iteration of the map $\tau$ and this is sufficient to guarantee that the limit point of the resulting sequence is an optimizer for the dual problem \eqref{eq:dual_statements}, as stated in the following Theorem. 

\begin{theo}[Convergence of the non-commutative Sinkhorn algorithm]\label{theo:introsinkhorn}
	Fix $\eps>0$. The definition \eqref{eq:def_implicit-T_i_intro} of the operators $\Sink_i$ is well-posed. Additionally, for any initial matrix $\bm U^{(0)} = (U_1, \dots , U_N) \in \bigtimes_{j=1}^N \cS^{\hat d_j}$, there exist $\bm \alpha^k \in \bR^{N}$ with $\sum_{i=1}^N \bm \alpha^k_i =0$ such that
	\begin{align}
		\bm U^{(k)}
		:= \tau^k( \bm U^{(0)}) + \bm \alpha^k \rightarrow  {\bm U}^\ep \quad \text{as } k \to +\infty ,
	\end{align}
	where $ { \bm U }^\ep = (  U_1^\ep, \dots,   U_N^\ep)$ is optimal for the dual problem and $\tau^k:= \underbrace{\tau \circ \dots \circ \tau}_{k\text{-times}}$.
	
	\noindent
	Consequently, if one defines for  $k \in \bN$
	\begin{align}
		\Gamma^{(k)}:=\exp\left(\frac {\bigoplus_{i=1}^N (\bm U^{(k)})_i-\H_{\mathcal O}}{\ep}\right) \quad \text{on } \mathcal O,
	\end{align}
	and $0$ on $\mathcal O^{\perp}$, then $\Gamma^{(k)}\to  {\Gamma^{\ep}}$ as $k \to +\infty$ where $ { \Gamma^{\ep}}$, $ { \bm U }^\ep$ satisfy \eqref{eq:optimality_maintheorem}. In particular, $  \Gamma^{\ep}$ is optimal for $\primin(\bm \gamma)$.
\end{theo}


\begin{rem}[Renormalisation]
	In the previous theorem, a renormalisation procedure is needed in order to obtain compactness for the dual potentials $\bm U^k$. Nonetheless, due to the fact that $\sum_{i=1}^N (\bm \alpha^k)_i = 0$ and by the properties of the operator $\bigoplus$, we observe that for $ k \in \bN$, the equality
	\begin{align*}
		\Gamma^{(k)} = \exp\left(\frac {\bigoplus_{i=1}^N  \tau^k(\bm U)_i-\H_{\mathcal O}}{\ep}\right) 
		\quad 	\text{on } \mathcal O
	\end{align*}
	is also satisfied. In fact, this shows that no renormalisation procedure is needed at the level of the primal problem, i.e. for the density matrices $\Gamma^{(k)}$.
\end{rem}

\begin{rem}(Umegaki Relative entropies)
	Similar results can be obtained if instead of the Von Neumann entropy one uses the quantum Umegaki relative entropy with respect to a reference density matrix with trivial kernel. Specifically, suppose that $m_i\in \cS^{d_i}$ with $\ker m_i = \{ 0 \}$. Then one can consider the minimisation problem
	\begin{align*}
		\prim_{\bm m}(\bm \gamma) = \inf \left\lbrace \Tr(\H\Gamma)+ \eps  S(\Gamma| \bm m) \suchthat \Gamma\in\DM^{\bm d} \text{ and } \Gamma\mapsto \bm \gamma  \right\rbrace ,
	\end{align*}
	where we set $\bm m:= \bigotimes_{i=1}^N m_i$ and $S(\Gamma| \bm m):=\Tr(\Gamma( \log \Gamma- \log \bm m)$ denotes the relative entropy of $\Gamma$ with respect to $\bm m$. The functional $\prim$ defined in \eqref{eq:primal_statements} corresponds to the case $\bm m$ equals the identity matrix. The corresponding dual functional $\dual$ as defined in \eqref{eq:dual_statements} is replaced by
	\begin{align*}
		\dualin_{\bm m}(\bm \gamma) = \sup \Bigg\{  \sum_{i=1}^N \Tr(U_i \gamma_i)-\eps \Tr \bigg( \exp\bigg[ \frac{\bigoplus_{i=1}^N U_i - \H_{\bm m}^\eps}{\eps} \bigg] \bigg) \suchthat U_i\in \cS^{d_i} \Bigg\} + \eps, 
	\end{align*}
	for a modified matrix $\H_{\bm m}^\eps := \H -\eps \log \bm m$ (restricted to $\mathcal O$ in the case of non-trival kernels). It is easy to see that our approach can also be used in this case. In particular, performing a change of variables in the dual potentials of the form $\tilde {\bm U} = \bm U+  \eps \log \bm m$ and using that $S(\Gamma | \text{Id}) =  S(\Gamma | \bm m) + \sum^N_{i=1}\left[S(\gamma_i)-S(\gamma_i|m_i)\right] $, one readily derives the validity of the same results obtained in Theorem \ref{theo:duality} and Theorem \ref{theo:introsinkhorn}, with the substitution of $\H$ with $\H_{\bm m}^\eps$. 
\end{rem}

\subsection{The symmetric case: one-body reduced density matrix functional theory}\label{sec:symmetryresults}

We are able to obtain the duality results stated above also in the symmetric cases (either bosonic or fermionic). For given $d,N \in \bN$, we set $\bm d = d^N$. We consider the bosonic (resp. fermionic) projection operator $\Pi_+$ (resp. $\Pi_-$) 
\begin{align}\label{eq:projectionFermionicBosonic}
	\Pi_+: \bigotimes_{i=1}^N \bC^d \to \bigodot_{i=1}^N \bC^d  \, ,  \qquad 	\Pi_-:\bigotimes_{i=1}^N \bC^d \to \bigwedge_{i=1}^N \bC^d   \, ,
\end{align}
where $\odot$ (resp. $\wedge$) denotes the symmetric (resp. antisymmetric) tensor product. Note that the cardinality of $\bigwedge_{i=1}^N \bC^d$ is $\binom{d}{N}$, therefore $\bigwedge_{i=1}^N \bC^d\neq\{0\}$ if and only if $N \leq d$. We denote by 
\begin{align} \label{eq:bosonicfermionicDM}
	\DM_+^{\bm d}:=\DM\left(\bigodot_{i=1}^N \bC^d\right), \quad \DM_-^{\bm d}:=\DM\left(\bigwedge_{i=1}^N \bC^d\right),
\end{align}
the set of bosonic and fermionic density matrices. We fix  $\H \in \cS^{\bm d}$ such that 
\begin{align}	\label{eq:symmetry}
	\tS_i \circ \H \circ \tS_i = \H \, , \quad \forall i = 1, \dots, N \, ,
\end{align}
where the $\tS_i$ are the permutation operators in Definition \ref{def:permope}. It is well-known that there exists $\Gamma \in \DM_-^d$ such that $\Gamma \mapsto \gamma$ (where $\Gamma \mapsto \gamma$ means that $\Gamma$ has all marginals equal to $\gamma$) if and only if $\gamma$ satisfies the \textit{Pauli exclusion principle}, i.e. if and only if $\gamma \in \DM^d$ and $\gamma\leq 1/N$ (see for example \cite[Theorem 3.2]{LiebSeiringer10}).

\begin{defin}[Bosonic and fermionic primal problems]	\label{def:bosons-fermions-primal}
	For any $\gamma \in \DM^d$, we define the \textit{bosonic} primal problem as 
	\begin{align}	\label{eq:primal_bosons}
		\prim_+(\gamma) := \inf \left\lbrace \Tr(\H\Gamma)+ \ep  \Tr(\Gamma \log \Gamma) \suchthat \Gamma\in\DM_+^{\bm d} \text{ and } \Gamma\mapsto \gamma  \right\rbrace \, .
	\end{align}
	For any $\gamma \in \DM^d$ such that $\gamma\leq 1/N$, we define the \textit{fermionic} primal problem as 
	\begin{align}	\label{eq:primal_fermions}
		\prim_-(\gamma) := \inf \left\lbrace \Tr(\H\Gamma)+ \ep  \Tr(\Gamma \log \Gamma) \suchthat \Gamma\in\DM_-^{\bm d} \text{ and } \Gamma\mapsto \gamma  \right\rbrace \,.
	\end{align}
	
\end{defin}
An analysis of the extremal points and the existence of the minimizer in \eqref{eq:primal_bosons} and \eqref{eq:primal_fermions} have been carried out in \cite{Col63} for the zero temperature case, and in \cite{GieRug19} in the positive temperature case.
\noindent
As in the non-symmetric case, we consider the associated bosonic and fermionic dual problems. For any given operator $A \in \cS^{\bm d}$, we denote by $A_\pm$ the corresponding projection onto the symmetric space, obtained as $A_\pm := \Pi_\pm \circ A \circ \Pi_\pm$ .

\begin{defin}[Bosonic and fermionic dual problems]	\label{def:bosons-fermions-dual}
	For any $\gamma \in \DM^d$, we define the \textit{bosonic}  dual functional $\fbD$ and the \textit{fermionic} dual functional $\ffD$ as 
	\begin{align}	\label{eq:primal_fermions_bosons}
		\fbfD: \cS^d \to \R \, , \quad \fbfD(U) := \Tr( U \gamma ) - \ep
		\Tr \left( 
		\exp \bigg[ \frac1\ep \bigg(\frac1N  \bigoplus_{i=1}^N U - H \bigg)_\pm  \bigg]
		\right) + \ep \,.
	\end{align}
	The corresponding dual problems are given by
	\begin{align}	\label{eq:dual_fermions_bosons}
		\bfD(\gamma) := \sup \left\{ \fbfD(U) \suchthat U \in \cS^d \right\} \, .
	\end{align}
\end{defin}

We note that a priori $\fD(\gamma)$ can be defined for any $\gamma \in \DM^d$, whereas $\prim_-(\gamma)$ is only well defined for $\gamma\in \DM^d$ such that $\gamma\leq 1/N$. This constraint on the primal problem naturally translates to an admissibility condition in order to have $\fD(\gamma) < \infty$. To ensure the existence of a maximizer for $\ffD$ we further need to impose $\gamma<1/N$. 
The following proposition gives an interesting and variational point of view of the Pauli principle, and it is proved in Section \ref{sec:dual Pauli}.
\begin{prop}[Pauli's principle and duality]		\label{prop:pauli}
	We have the following equivalences:
	\begin{enumerate}
		\item $
		\fD(\gamma) < \infty
		$ if and only if $\gamma \in \DM^d$ and $\gamma \leq \frac1N$, 
		\item There exists a maximiser $U_0 \in \cS^d$ of $\emph{D}_\gamma^{-,\eps}$ if and only if $\gamma \in \DM^d$ and  $0 < \gamma < \frac1N$. 
	\end{enumerate}
\end{prop}

Finally we state the duality result in the fermionic and bosonic setting.

\begin{theo}[Fermionic and bosonic duality] \label{theo:dualBosFerm}
	Let $\H \in \cS^{\bm d}$ satisfying \eqref{eq:symmetry}. 
	\begin{itemize}
		\item[(i)]  For any given $\gamma \in \DM^d$, such that $\gamma \leq \frac1N$, the fermionic primal and dual problems coincide, thus $\prim_-(\gamma) = \fD(\gamma)$. Moreover, if $0 < \gamma < \frac1N$ then $\emph{D}_\gamma^{-,\eps}$ admits a unique maximizer $U_-^\ep$ such that
		\begin{align}	\label{eq:optimum_primal_fermions}
			\Gamma_-^\ep = \exp \left( \frac1\ep \bigg[ \frac1N \bigoplus_{i=1}^N U_-^\ep - \H  \bigg]_- \right) 
		\end{align}
		is the unique optimal fermionic solution to the primal problem $\prim_-(\gamma)$.
		\item[(ii)] For any given $\gamma \in \DM^d$, the bosonic primal and dual problems coincide, thus $\prim_+(\gamma) = \bD(\gamma)$. Moreover, if $\gamma >0$, $\emph{D}_\gamma^{+,\eps}$ admits a unique maximizer $U_+^\ep$ such that
		\begin{align}	\label{eq:optimum_primal_bosons}
			\Gamma_+^\ep = \exp \left( \frac1\ep \bigg[ \frac1N \bigoplus_{i=1}^N U_+^\ep - \H  \bigg]_+ \right) 
		\end{align}
		is the unique optimal bosonic solution to the primal problem $\prim_+(\gamma)$.
	\end{itemize}
\end{theo}

\section{Preliminaries and a priori estimates}\label{sec:prem}

We start this section by recalling the setting and the notation. For $d \in \N$, we denote by $\cM^d = \cM^d(\bC)$ the set of all $d \times d$ complex matrices, by $\cS^d$ the hermitian elements of $\cM^d$, and by $\cS_\geq^d$ (respectively $\cS_>^d$) the set of all the positive semidefinite (positive definite) elements of $\cS^d$. With a slight abuse of notation, we denote by $\Tr$ the trace operator on $\cM_d$ for any dimension $d$.  Furthemore, we denote by $\DM^d$ the set of $d\times d$ \textit{density matrices}, namely the matrices in $\cS_\geq^d$ with trace one. For the sake of notation, for every $N \in N$ we denote by $[N]:= \{1, ... , N\}$.

For a given $N \in \N$ and $(d_i)_{i=1}^N \subset \N$, we consider for any $i \in [N]$ the injective maps
\begin{gather}	\label{eq:defQ_i}
	\begin{gathered}
		Q_i : \cM^{d_i} \to \cM^{\bm d} = \bigotimes_{j=1}^N \cM^{d_j}, \quad \bm d:= \prod_{j=1}^N d_j, \\
		\forall A \in \cM^{d_i}, \quad
		\tQ_i(A) := \bigotimes_{j=1}^N A_j, \quad  A_j = \begin{cases}
			A & \text{if } j = i, \\
			\unit & \text{if } j \neq i.
		\end{cases} 
	\end{gathered}
\end{gather}
We shall use the same notation also for subsets of $\bC^d$. I.e., we also denote by $\tQ_i$ the map $\tQ_i : \bC^{d_i} \to \bC^{\bm d}$ defined as
\begin{align*}
	\forall K \subset \bC^{d_i}, \quad \tQ_i(K) := \bigotimes_{j=1}^N K_j \subset \bC^{\bm d},  \quad K_j = \begin{cases}
		K & \text{if } j = i, \\
		\bC^{d_j} & \text{if } j \neq i.
	\end{cases} 
\end{align*}
The \textit{marginal} operators are the left-inverse of the $\tQ_i$, namely 
$\tP_i :  \cM^{\bm d} \to \cM^{d_i}$,
where for every $\Gamma \in \cM^{\bm d}$, $\tP_i(\Gamma) \in \cM^{d_i}$ is defined by duality as
\begin{align}	\label{eq:defP_i}
	\Tr( \tP_i(\Gamma) A ) = \Tr \Big( \Gamma  \tQ_i (A) \Big), \quad \forall A \in \cM^{d_i}.
\end{align}

\begin{rem}	\label{rem:prop_partialtrace}
	Observe that $\Tr (\tP_i(A)) = \Tr(A)$ for every $i=1, \dots , N$ and $A \in \cM^{\bm d}$. Furthermore, if $A=\bigotimes_{i=1}^N A_i$ with $\Tr(A_i)= 1$, then $\tP_i(A) = A_i$. 
\end{rem}
For a given family of density matrices $\gamma_i \in \DM^{d_i}$ , we use the notation $\bm \gamma := (\gamma_i)_{i \in [N]}$ and we write $\Gamma \mapsto \bm \gamma=(\gamma_1, \dots, \gamma_N)$ whenever $\Gamma \in \DM^{\bm d}$ and $\tP_i (\Gamma) = \gamma_i$ for every $i=[N]$.
With the next definitions, we introduce the Kronecker sum and Permutation operators.
\begin{defin}[Kronecker sum]
	For $A_i \in \cM^{d_i}$, we call their \textit{Kronecker sum} the matrix 
	\begin{align}	\label{eq:def_Kronecker}
		\bigoplus_{i=1}^N A_i:= \sum_{i=1}^N \tQ_i(A_i) \in \cM^{\bm d}
	\end{align}
	where $\tQ_i$ is defined in \eqref{eq:defQ_i}.
\end{defin}
%

\begin{defin}[Permutation operators]\label{def:permope}
	For any $i \in [N]$, we introduce the \textit{permutation operator} $\tS_i : \cM^{\bm d} \approx \bigotimes_{j=1}^N \cM^{d_j} \to \cM^{\bm d}$ as the map defined by 
	\begin{align*}
		\tS_i \bigg( \bigotimes_{j=1}^N A_j \bigg) = A_1 \otimes \dots \otimes A_{i-1} \otimes A_{i+1} \otimes \dots \otimes A_N \otimes A_i,
	\end{align*}
	for any $A_i \in \cM^{d_i}$ and extended to the whole $\cM^{\bm d}$ by linearity.
\end{defin}

\begin{rem}	\label{rem:proptSi}
	The permutation operators preserve the spectral properties of any operator. Precisely, $\sigma(\tS_i(A)) = \sigma(A)$ for every $i \in [N]$, $A \in \cS^{\bm d}$, where $\sigma(A)$ denotes the spectrum of $A$.
	In particular, for every continuous function $f: \bR \to \bR$, we have that $\Tr (f(\tS_i(A))) = \Tr(f(A))$, for every $A \in \cS^{\bm d}$.
\end{rem}

\subsection{Non-commutative $(\H,\ep)$-transforms}	\label{sec:Heptransform}

For this section, we specify to the simply case of a two-fold tensor product and introduce the notion of non-commutative $(\H,\ep)$-transform, which is a central object in our discussion. We shall see in Section \ref{sec:vectorial} how it is then easy to extend this notion to a general $N$-fold tensor product. We fix $d,d' \in \bN$, $0<\alpha \in \DM^{d'}$, $\H \in \cS^{dd'}$ and $\ep >0$ and define the map $\tT_{\alpha,\H}^\ep: \cS^d \times \cS^{d'} \to \bR $ as
\begin{align}	\label{eq:functionaltransformone}
	\tT_{\alpha,\H}^\ep(U,V):= \Tr (V \alpha) - \ep \Tr \left(\exp\left[\frac{U\oplus V-\H}{\ep}\right]\right).
\end{align}
The $(\H,\ep)$-transform of any $U \in \cS^d$ is obtained as the maximiser of the map $\tT_{\alpha,\H}^\ep(U, \cdot)$.

\begin{defin}[$(\H,\ep)$-transform]
	We call the unique maximizer of $\tT_{\alpha,\H}^\ep(U, \cdot)$ the \textit{$(\H,\ep)$-transform} of $U \in \cS^d$. We use the notation 
	\begin{align}
		\fT_{\alpha, \H}^\ep: \cS^d \to \cS^{d'}, \quad  \fT_{\alpha, \H}^\ep(U) = \argmax \{ \tT_{\alpha,\H}^\ep(U, V) \suchthat V \in \cS^{d'}\}.
	\end{align}
\end{defin}

The following lemma shows that the definition of $(\H,\ep)$-transform is indeed well-posed.
\begin{lemma}	\label{lem:existence of Hep transform}
	Let $U \in \cS^d$. Then there exists a unique maximizer $\bar{V}\in \cS^{d'}$ of $\tT_{\alpha,\H}^\ep(U, \cdot)$.
\end{lemma}

\begin{proof}
	Fix $U \in \cS^{d}$. For every $V\in \cS^{d'}$, we write $V = V_+ - V_-$ where $V_+$,$V_- \in \cS^{d'}$ denote respectively the positive and the negative part of $V$ (with respect to its spectrum). 
	We begin by observing that
	\begin{align}  \label{eq:lowerbounddoubtrace}
		\begin{gathered}
			\Tr\left(\exp\left[\frac{U\oplus V-\H}{\ep}\right]\right)\geq \Tr\left(\exp\left[\frac{U\oplus V-\|\H\|_{\infty}}{\ep}\right]\right)\\
			=\Tr\left(\exp\left(\frac V {\ep}\right)\right) \Tr \left(\exp\left(\frac U {\ep}\right)\right) \exp\left[\frac{-\|\H\|_{\infty}}{\ep}\right]=:\kappa\Tr\left(\exp\left(\frac V{\ep}\right)\right)\geq \kappa e^{\ep^{-1}\| V_+ \|_\infty},
		\end{gathered}
	\end{align}
	where in the second step we used that $\exp(U\oplus V)=\exp(U)\otimes\exp(V)$ and $\kappa=\kappa(U,\ep,\H)$ is a \emph{finite} constant depending on $U$, $\ep$, and $\H$. On the other hand, it clearly holds $\Tr(V \alpha )\leq \| V_+ \|_\infty$ which combined with \eqref{eq:lowerbounddoubtrace} yields for every $V \in \cS^{d'}$
	\begin{align}
		\label{eq:Gbound1}
		\tT_{\alpha,\H}^\ep(U,V)\leq \| V_+ \|_\infty-\kappa e^{\ep^{-1} \| V_+ \|_\infty}.
	\end{align}
	Moreover, it is immediate to obtain that 
	\begin{align}
		\label{eq:Gbound2}
		\begin{aligned}
			\tT_{\alpha,\H}^\ep(U,V) \leq \Tr(V\alpha)
			=  \Tr(V_+ \alpha) - \Tr(V_- \alpha) \leq \| V_+ \|_\infty -\sigma_{\min}(\alpha)  \| V_- \|_\infty,
		\end{aligned}
	\end{align}
	where $\sigma_{\min}(\alpha)$ is the spectral gap of $\alpha$, which is strictly positive by assumption.
	Let $V_n$ be a maximizing sequence for $\tT_{\alpha,\H}^\ep(U,\cdot)$, then the bounds \eqref{eq:Gbound1} and \eqref{eq:Gbound2} imply that $ (V_n)_+ $, $ (V_n)_-$ (and hence $V_n$) are uniformly bounded. Therefore, we can obtain a subsequence (which we do not relabel) such that
	$
	V_n \rightarrow \bar{V} \in \cS^{d'}.
	$
	The optimality of $\bar{V}$ follows from the fact that  $\tT_{\alpha,\H}^\ep(U,\cdot)$ is continuous and strictly concave (see for example \cite{carlen2010trace}), which also implies uniqueness. 
\end{proof}

In the following lemma we use the fact that the $(\H,\ep)$-transform is obtained through a maximization to show that it can be characterized as the solution of the associated Euler--Lagrange equation. This property is crucial for the proof of our main results.

\begin{lemma}[Optimality conditions for the $(\H,\ep)$-transforms]	\label{lemma:Ucep}
	Given $d,d' \in \bN$, $0<\alpha \in \DM^{d'}$, $\H \in \cS^{dd'}$, $\ep>0$, the operator  $\fT_{\alpha, \H}^\ep$ can be characterized implicitly by the fact that, for any $U \in \cS^d$, $\fT_{\alpha, \H}^\ep(U)$ is the unique solution of
	\begin{align}		\label{eq:implicit_Ctransform}
		\alpha = \tP_2 \left(\exp\left[\frac{U \oplus \fT_{\alpha, \H}^\ep(U)-\H}{\ep}\right]\right).
	\end{align}
\end{lemma}
\begin{proof}
	Let us pick any $\Lambda \in \cS^d$ and define $V_s := \fT_{\alpha, \H}^\ep(U) + s \Lambda$. By construction, due to the optimality of $\fT_{\alpha, \H}^\ep(U)$, the map
	\begin{align*}	
		s \mapsto g(s):=\Tr (V_s \alpha) - \ep \Tr\left(\exp\left[\frac{U\oplus V_s-\H}{\ep}\right]\right) 
	\end{align*}
	must have vanishing derivative at $s=0$. This can be computed \cite[Section 2.2]{carlen2010trace} as
	\begin{align} \label{eq:firstv}
		g'(0)= \Tr (\Lambda \alpha) - \ep \Tr\left((I \otimes \Lambda) \exp\left[\frac{U \oplus \fT_{\alpha, \H}^\ep(U)-\H}{\ep}\right]\right).
	\end{align}
	Using the definition of partial trace and the previous formula, we infer
	%
	%
	\begin{align}	\label{eq:weak_EL_ctransform}
		\Tr \left( \Lambda \Big( \alpha - \tP_1 \Big(\exp\left[\frac{U\oplus \fT_{\alpha, \H}^\ep(U)-\H}{\ep}\right]\Big) \Big) \right) = 0
	\end{align}
	for every $\Lambda \in \cS^d$. Note that $\alpha,U,V_s,\H$ being self-adjoint, it follows that the operator 
	\begin{align*}
		\alpha - \tP_1 \left(\exp\left[\frac{U\oplus \fT_{\alpha, \H}^\ep(U)-\H}{\ep}\right]\right)
	\end{align*}
	is self-adjoint as well. Together with \eqref{eq:weak_EL_ctransform}, 
	this shows \eqref{eq:implicit_Ctransform}. On the other hand, since \eqref{eq:weak_EL_ctransform} is the Euler Lagrange equation associated to the maximization of the strictly concave functional $\tT_{\alpha,\H}^\ep(U,\cdot)$, any solution of \eqref{eq:weak_EL_ctransform} is necessarily a maximizer and hence coincides with $\fT_{\alpha, \H}^\ep(U)$, by uniqueness (see Lemma \ref{lem:existence of Hep transform}).
	
\end{proof}
%
%
The next step is to obtain some regularity estimates on $\fT_{\alpha, \H}^\ep(U)$. To do so, we extrapolate information from the optimality conditions proved in Lemma \ref{lemma:Ucep}. 

%
\begin{proposition}[Regularity of the $(\H,\ep)$-transform] \label{prop:bounds_Ctransforms} 
	Given $d,d' \in \bN$, $0<\alpha \in \DM^{d'}$, $\H \in \cS^{dd'}$, $\ep>0$, we define for all $A \in \cS^d$ (or $A \in \cS^{d'}$)
	\begin{align}	\label{eq:def_lambda}
		\lambda_\ep(A):= \ep \log \Big( \Tr \Big[ \exp \left( \frac{A}{\ep} \right) \Big] \Big).
	\end{align}
	Then for every $U \in \cS^d$  it holds
	\begin{gather}
		\label{eq:reg_1}
		\Big| \fT_{\alpha, \H}^\ep(U) - \ep \log \alpha  + \lambda_\ep(U) \unit \Big| \leq \Vert {\rm H}\Vert_{\infty} \unit,\\
		\label{eq:control_lambda}
		\Big| \lambda_\ep(U) + \lambda_\ep \big( \fT_{\alpha, \H}^\ep(U) \big) \Big| \leq \|{\rm H}\|_{\infty},\\ 
		\label{eq:control_transflambda}
		\Big| \fT_{\alpha, \H}^\ep(U) - \ep \log \alpha -\lambda_\ep (\fT_{\alpha, \H}^\ep(U))\unit\Big|\leq 2\Vert {\rm H} \Vert_{\infty} \unit.
	\end{gather}
	where the inequalities are understood as two-sided quadratic forms bounds.
\end{proposition}

\begin{proof}
	Note that \eqref{eq:control_transflambda} is an immediate consequence of \eqref{eq:reg_1} and \eqref{eq:control_lambda} and we shall therefore only prove the latter two. Let us start with the proof of \eqref{eq:reg_1}. 
	We know from Lemma \ref{lemma:Ucep} that for every $U \in \cS^d$, $\fT_{\alpha, \H}^\ep(U)$ satisfies equation \eqref{eq:implicit_Ctransform}. By the properties of the partial trace (Remark \ref{rem:prop_partialtrace}) and $\H \leq \| {\rm H} \|_{\infty} \unit$, it follows that 
	\begin{align}	\label{eq:proof_upb_reg}
		\begin{aligned}
			\alpha 
			&\leq e^{\frac{\| {\rm H} \|_{\infty}}{\ep}} \tP_1 \left(\exp\left[\frac{U \oplus \fT_{\alpha, \H}^\ep(U) }{\ep}\right]\right) \\
			&= e^{\frac{\| {\rm H} \|_{\infty}}{\ep}} \tP_1 \left( \exp\left( \frac U{\ep} \right) \otimes \exp \left( \frac{\fT_{\alpha, \H}^\ep(U)}{\ep} \right) \right) \\
			&= e^{\frac{\| {\rm H} \|_{\infty}}{\ep}}  \Tr \left( \exp\left( \frac{U }{\ep} \right) \right)  \exp \left( \frac{\fT_{\alpha, \H}^\ep(U)}{\ep} \right) \, ,
		\end{aligned}
	\end{align}
	where in the first inequality we used that $\exp(A \oplus B) = \exp A \otimes \exp B$.
	%
	%
	Similarly, using instead the lower bound $\H \geq - \| {\rm H} \|_{\infty} \unit$, from \eqref{eq:implicit_Ctransform} we can also obtain 
	\begin{align}	\label{eq:proof_lbd_reg}
		\alpha \geq e^{\frac{-\| {\rm H} \|_{\infty}}{\ep}}  \Tr \left( \exp\left( \frac{U}{\ep} \right) \right)  \exp \left( \frac{\fT_{\alpha, \H}^\ep(U)}{\ep} \right).
	\end{align}
	We can put together the two bounds in \eqref{eq:proof_upb_reg}, \eqref{eq:proof_lbd_reg} to obtain
	\begin{align}	\label{eq:proof_expbound_reg}
		\alpha e^{\frac{-\| {\rm H} \|_{\infty}}{\ep}} \leq \Tr \left( \exp\left( \frac{U  }{\ep} \right) \right) \exp \left( \frac{\fT_{\alpha, \H}^\ep(U) }{\ep} \right) \leq \alpha e^{\frac{\| {\rm H} \|_{\infty}}{\ep}}.
	\end{align}
	Taking the log in the latter inequalities we conclude the proof of \eqref{eq:reg_1}. If we instead take the trace of both sides in \eqref{eq:proof_expbound_reg}, we obtain
	
	\begin{align*}
		e^{\frac{-\| {\rm H} \|_{\infty}}{\ep}} \leq \Tr \left( \exp\left( \frac{U}{\ep} \right) \right) \Tr \left( \exp \left( \frac{\fT_{\alpha, \H}^\ep(U)}{\ep} \right) \right) \leq e^{\frac{\| {\rm H} \|_{\infty}}{\ep}},
	\end{align*}
	and then applying the log, we conclude the proof $\eqref{eq:control_lambda}$.
\end{proof}

\subsection{Vectorial $(\H,\ep)$-transforms}
\label{sec:vectorial}

In this section, we consider a vectorial version of the $(\H,\ep)$-transforms introduced in the previous section. This turns out to be a key object in the proof of Theorem \ref{theo:duality} and Theorem \ref{theo:introsinkhorn}, necessary to deal with the multi-marginal setting.
%

Let us first introduce the general framework, which remains in force throughout the section.
Let $N \in \bN$ and $[N]$ be a index set of $N$ elements. For all $i \in [N]$, let $d_i\in \N$ and $\gamma_i \in \DM^{d_i}$ be density matrices. Set $\bm\gamma := ( \gamma_i )_{i \in [N]}$, $\bm d = \prod_{j=1}^N d_i$. Finally, consider a Hamiltonian $\H \in \cS^{\bm d}$.
\begin{rem}(Kernels)	\label{rem:kernels}
	Without loss of generality, we can assume $\ker \gamma_i = \{ 0 \}$, for every $i \in [N]$. In the general case, it suffices to consider the restriction to the set $	\mathcal O:= \bigotimes_{i=1}^N \big( \ker \gamma_i \big)^\perp $ and consider the matrix $\H_{\mathcal O}= \Pi_{\mathcal O} \H \Pi_{\mathcal O}$, where $\Pi_{\mathcal O}$ is the projector onto $\mathcal O$. 
\end{rem}

%
%
We therefore assume that $\ker \gamma_i = \{ 0 \}$ for all $i\in [N]$. 
In this section we extend the notion of $(\H,\ep)$-transform as introduced in previous section \ref{sec:Heptransform} to the multi-marginal setting, and we apply it to our specific setting. We are interested in the maximization \eqref{eq:dual_statements} of the dual functional, that we introduce below.
\begin{defin}[Dual Functional] \label{eq:dualfcfull}
	For any $\bm U=(U_1,\dots,U_N)\in  \bigtimes_{j=1}^N \cS^{d_j}$, we define
	\begin{align*}
		\dualf(\bm U)= \sum_{i=1}^N \Tr(U_i \gamma_i)-\ep \Tr \bigg( \exp\bigg[ \frac{\bigoplus_{i=1}^N U_i - \H}{\ep} \bigg] \bigg) + \ep.
	\end{align*}
	
\end{defin}

\begin{rem} 
	\label{rem:dualinvariancebytransl}
	Note that $\dualf$ is invariant by translation for any vector $\bm a=(a_1,\dots, a_N)\in \bm \R^N$ such that $\sum_{k=1}^N a_k=0$, i.e. 
	\begin{align*}
		\dualf(\bm U+\bm a)=\dualf(\bm U).
	\end{align*}
	As a consequence of this property, we see in Section  \ref{sec:noncomsin} that the set of maximizers is invariant by such transformations (Lemma \ref{lem:characterizionmax}).
\end{rem}

With the following definition, we introduce the vectorial $(\H,\ep)$-transforms.

\begin{defin}[Vectorial $(\H,\ep)$-transform]	\label{def:vecHep}
	For any $i \in [N]$, we define the \textit{$i$-th vectorial $(\H,\ep)$-transform} $\fT_i^\ep$ as the map
	\begin{gather*}
		\begin{gathered}
			\fT_i^\ep : \bigtimes_{j=1, \, j \neq i}^N \cS^{d_j} \to \cS^{d_i}, \\
			\fT_i^\ep(\hat{\bm U}_i) = \Argmax_{V \in \cS^{d_i}} \left\{  \Tr(V \gamma_i)-\ep\Tr\left(\exp\left[\frac1\ep \left(  U_1 \oplus \dots \oplus U_{i-1} \oplus V \oplus U_{i+1} \oplus \dots \oplus U_N-\H \right) \right]\right)   \right\},
		\end{gathered}
	\end{gather*}
	where for $\bm U \in \bigtimes_{j=1}^N \cS^{d_j}$, we set $\hat{\bm U}_i  $ to be the product of all the $U_j$ but the $i$-th one, namely
	\begin{align}	\label{eq:def_hatU_i}
		\hat{\bm U} _i :=
		\big(
		U_1, \dots, U_{j-1}, U_{j+1}, \dots, U_N	
		\big)
		\in \bigtimes_{j=1, \, j \neq i}^N \cS^{d_j}.
	\end{align} 
\end{defin}

\begin{rem}
	\label{rem:vecttransfabstracttransf}
	Observe that we can identify the $i$-th vectorial $(\H,\ep)$-transforms with a particular case of the operators $\fT_{\ep,\H, \alpha}$ as introduced in Section \ref{sec:Heptransform}. Indeed, as a consequence of Remark \ref{rem:proptSi} it is straightforward to see that for $i \in [N]$ 
	\begin{align}	\label{eq:multim_transform_twom}
		\fT_i^\ep(\hat{\bm U}_i) = \fT_{\gamma_i,\tS_i(\H)}^\ep\left( \bigoplus_{j=1, \, j \neq i}^N U_i  \right), \quad 
		\fT_{\gamma_i,\tS_i(\H)}^\ep: \bigotimes_{j=1, \, j \neq i}^N \cS^{d_j} \approx \cS^{\tilde d_i} \to \cS^{d_i},
	\end{align}
	where we set $\tilde d_i := \prod_{j\neq i} d_j$ and the $\tS_i$ are the permutation operators in Definition \ref{def:permope}. This shows that the definition is well posed (i.e. that the $\Argmax$ appearing in the definition exists and is unique). Moreover it allows us to extend the validity of the properties of the $(\H,\ep)$-transform shown in Section \ref{sec:Heptransform} to the operators $\fT_i^\ep$, as we shall see in Lemma \ref{lemmaN:Ucep} and Proposition \ref{prop:reg_Hep-transf} below. Note that the dependence on the specific entry $i$ is reflected in both the use of $\gamma_i$ and in the fact that the transform is performed w.r.t. $\tS_i(\H)$.
\end{rem}

%


\subsection{One-step and Sinkhorn operators}

We use the vectorial $(\H, \ep)$-transforms to define what we call \textit{one-step operators} and \textit{Sinkhorn operators}. The first ones map a vector of $N$ potentials into a vector of $N$ potentials, exchanging its $i$-th entry with the $i$-th vectorial $(\H,\ep)$-transform applied to the other $N-1$. The second is simply obtained by composing all the different $N$ one-step operators.

\begin{defin}[One-step operators]	\label{def_Tiep}
	For $i \in [N]$, we introduce the \textit{one-step operators}  $\Sink_i$, which are defined by 
	\begin{gather*}
		\Sink_i: \bigtimes_{j=1}^N \cS^{d_j} \to \bigtimes_{j=1}^N \cS^{d_j}\\ 
		\bm U:= (U_1, \dots, U_N) \longmapsto (U_1,\dots,U_{i-1}, \fT_i^\ep(\hat{\bm U}_i),U_{i+1}, \dots, U_N)=:\Sink_i(\bm U).
	\end{gather*} 
\end{defin}

The Sinkhorn operator is simply the composition of the $N$ one-step operators $\Sink_i, i \in [N]$.

\begin{defin}[Sinkhorn Operator] \label{def:SinkhornOp}
	We introduce the \textit{Sinkhorn operator} $\tau$, defined by
	\begin{align*}
		\begin{gathered}
			\tau: \bigtimes_{j=1}^N \cS^{d_j} \to \bigtimes_{j=1}^N \cS^{d_j} , \\
			\tau(\bm U) := (\Sink_N \circ \dots \circ \Sink_1) (\bm U).
		\end{gathered}
	\end{align*}
\end{defin} 

\begin{rem} \label{rem:maximizersfixed}
	Note that, by definition of $\tau$, it follows immediately that, for any $\bm U \in \bigtimes_{j=1}^N \cS^{d_j}$ 
	\begin{align*}
		\dualf(\tau(\bm U))\geq \dualf(\bm U),
	\end{align*}
	i.e. applying $\tau$ to any vector increases its energy. Moreover, any maximizer of $\dualf$ is a fixed point of $\tau$ (as a consequence of the uniqueness proved in Lemma \ref{lem:existence of Hep transform}). The converse is also true and implies that the set of maximizers of $\dualf$ coincides with the set of fixed points of $\tau$, see Remark \ref{rem:nonCommSch_multim}.
	%
\end{rem}

\begin{rem}
	\label{rem:translinvariance}
	Note that for any vector $\bm a\in \bm \R^N$ such that $\sum_{k=1}^N \bm a_k=0$, one has
	\begin{align*}
		\Sink_i(\bm U + \bm a)=\Sink_i(\bm U)+\bm a,
	\end{align*}
	i.e. $\Sink_i$ commutes with translations by vectors whose coordinates sum up to zero (notice that this fact is particularly interesting in light of Remark \ref{rem:dualinvariancebytransl}). This is a straightforward consequence of the fact that
	\begin{align*}
		\fT_i^\ep\left((\widehat{\bm U+\bm a})_i\right)=\fT_i^\ep(\hat{\bm U} _i)+a_i,
	\end{align*}
	which can be readily verified from the definitions.
	Trivially, this also implies
	\begin{align*}
		\tau(\bm U + \bm a)=\tau(\bm U)+\bm a.
	\end{align*}
\end{rem}

We now take advantage of the observations in Remark \ref{rem:vecttransfabstracttransf} to deduce properties for the vectorial $(\H,\ep)$-transforms, the one-step operators, and the Sinkhorn operator. First of all, as a corollary of Lemma \ref{lemma:Ucep}, we characterize the vectorial $(\H,\ep)$-transforms as solutions of implicit equations.

\begin{lemma}[Optimality conditions for vectorial $(\H,\ep)$-transforms]	\label{lemmaN:Ucep}
	Let $i\in [N]$, $\ep>0$, $\gamma_i \in \DM^{d_i}$, $\H \in \cS^d$, with $\ker \gamma_i = \{ 0 \}$. For any $\bm U \in \bigtimes_{j=1}^N \cS^{d_j}$, the one step-operator $\Sink_i(\bm U)$ (or equivalently the $i$-th vectorial $(\H,\ep)$-transform $\fT_i^\ep(\hat{\bm U}_i)$) is implicitly characterized as the unique solution of the equation
	\begin{align}	
		\label{eqN:implicit_CNtransform_i}
		\gamma_i = \tP_i \left(\exp\left[\frac1\ep \left( \bigoplus_{j=1}^N (\Sink_i(\bm U))_j - \H \right) \right]\right) \, .
	\end{align}
\end{lemma}

\begin{proof}
	As a consequence of \eqref{eq:multim_transform_twom}, we can apply Lemma \ref{lemma:Ucep} and deduce
	\begin{align*}
		\gamma_i
		&= \tP_i \left(\exp\left[\frac1\ep \left( \left( \bigoplus_{j=1, \, j \neq i}^N U_i \right) \oplus \fT_i^\ep(\hat{\bm U}_i)  - \tS_i(\H) \right) \right]\right)\\
		&=\tP_i \left(\exp\left[\frac1\ep \left( \bigoplus_{j=1}^N (\Sink_i(\bm U))_j - \H \right) \right]\right) \, ,
	\end{align*}
	where $\tS_i$ is the $i$-th permutation operator, as defined in \ref{def:permope}, and in the last equality we used Remark \ref{rem:proptSi} and that
	\begin{align*}
		\left( \bigoplus_{j=1, \, j \neq i}^N U_i \right) \oplus \fT_i^\ep(\hat{\bm U}_i)  = \tS_i \left( \bigoplus_{j=1}^N (\Sink_i(\bm U))_j \right)
	\end{align*}
	for every $i \in [N]$ and $\bm U \in \bigtimes_{j=1}^N \cS^{d_j}$.
\end{proof}
%
%

The next proposition collects the regularity properties of the $(\H,\ep)$-transforms. Once again, they are direct consequence of the properties proved in the two marginals case, in particular in Proposition \ref{prop:bounds_Ctransforms}. 
%
%
%
\begin{proposition}[Regularity of the $(\H,\ep)$-transforms]	\label{prop:reg_Hep-transf}	
	Let $i\in [N]$, $\ep>0$, $\gamma_i \in \DM^{d_i}$, $\H \in \cS^d$, with $\ker \gamma_i = \{ 0 \}$. Then for every $\bm U \in \bigtimes_{j=1}^N \cS^{d_j} $, for every $i \in [N]$ it holds
	\begin{gather}	
		\label{eq:reg_i}
		\left| \fT_i^\ep(\hat{\bm U}_i) - \ep \log \gamma_i   + \sum_{j=1, j\neq i}^N \lambda_\ep (U_j) \unit \right| \leq \| {\rm H} \|_{\infty} \unit \,,\\
		\label{eq:control_lambda_i}
		\left| \sum_{j=1, j\neq i}^N \lambda_\ep (U_j) + \lambda_\ep \left( \fT_i^\ep(\hat{\bm U}_i) \right) \right| \leq \| {\rm H} \|_{\infty}\, ,\\
		\label{eq:control_transflambda_i}
		\left|\fT_i^\ep(\hat{\bm U}_i) -\ep \log \gamma_i    -\lambda_\ep \left( \fT_i^\ep(\hat{\bm U}_i) \right)\unit\right|\leq (2\|{\rm H}\|_{\infty})  \unit, 
	\end{gather}
	where $\lambda_\ep$ is defined in \eqref{eq:def_lambda}.
\end{proposition}

\begin{proof}
	The proof is a direct application of Proposition \ref{prop:bounds_Ctransforms} and the considerations in Remark \ref{rem:vecttransfabstracttransf}. Precisely, the estimate \eqref{eq:reg_i} follows from \eqref{eq:reg_1}, \eqref{eq:control_lambda_i} follows from \eqref{eq:control_lambda} and \eqref{eq:control_transflambda_i} follows from \eqref{eq:control_transflambda}, together with the fact that
	\begin{align*}
		\lambda_\ep\left(\bigoplus_{j=1, \, j \neq i}^N U_j  \right)=\sum_{j=1, j\neq i}^N \lambda_{\varepsilon} \left( U_j\right)
		, \quad \forall \bm U \in \bigtimes_{j=1}^N \cS^{d_j}.
	\end{align*}
\end{proof}

In light of Remark \ref{rem:maximizersfixed}, it is reasonable to check whether sequences of the form $\tau^k(\bm U_0)$ are maximizing for $\dualf$ and compact. On the other hand, a priori it is not clear how to obtain compactness for such sequences and Remark \ref{rem:dualinvariancebytransl} shows that there could even exist sequences `converging' to the set of maximizers which are not compact. It is therefore natural to introduce a suitable renormalization operator, aimed at retrieving compactness. Note that any such operator should increase or leave invariant the value of $\dualf$ and therefore, by Remark \ref{rem:dualinvariancebytransl}, any translation by vectors whose coordinates sum up to zero is a good candidate.

\begin{defin}[Renormalisation]
	\label{def:Ren}
	Let $\lambda_\ep$ be defined as in \eqref{eq:def_lambda}. We define the \textit{renormalisation map} $\ren: \bigtimes_{i=1}^N \cS^{d_i} \to \bigtimes_{i=1}^N \cS^{d_i} $ as the function
	\begin{align*}
		\ren(\bm U)_i = \begin{cases}
			U_i - \lambda_{\ep}(U_i) , &\text{ if } i \in \lbrace 1,\dots,N-1\rbrace \\
			\displaystyle U_N + \sum_{j=1}^{N-1}\lambda_{\ep}(U_j), &\text{ if } i = N.\end{cases}
	\end{align*}
\end{defin}

In the following proposition we show that $\ren \left(\tau\left(\bigtimes_{i=1}^N \cS^{d_i}\right)\right)$ is bounded and therefore compact. This shows that the map $\ren$ is indeed a reasonable renormalization operator for our purposes. 
%
%
\begin{proposition}	[Renormalisation of $(\H,\ep)$-transforms and uniform bounds]	\label{propN:uniform_bounds}
	Let $i\in [N]$, $\ep>0$, $\gamma_i \in \DM^{d_i}$, $\H \in \cS^d$, with $\ker \gamma_i = \{ 0 \}$. Then, for any $\bm U  \in \cS^{\bm d}$, one has that 
	$
	\dualf(\ren \tau (\bm U) ) \geq \dualf(\bm U),
	$
	and the following bounds hold true:
	\begin{align}	\label{eqN:uniform_bound_Ui}
		\big| (\ren \tau(\bm U))_i - \eps \log \gamma_i   \big|	\leq 2 \| {\rm H} \|_{\infty}  \unit, \quad \forall i \in [N].	
	\end{align}
\end{proposition}

\begin{proof}	
	First of all,  Remark \ref{rem:dualinvariancebytransl} and Remark \ref{rem:maximizersfixed} trivially yield
	$
	\dualf(\ren \tau (\bm U) ) \geq \dualf(\bm U).
	$
	To show \eqref{eqN:uniform_bound_Ui}, note that for any $i\in [N]$, $(\tau(\bm U))_i$ is obtained applying $\fT_i^\ep$ to some element of $\bigtimes_{j=1, \, j \neq i}^N \cS^{d_j}$. Therefore, applying \eqref{eq:control_transflambda_i} from Proposition \ref{prop:reg_Hep-transf}, we obtain
	\begin{gather*}
		\left\|(\ren\tau(\bm U))_i -\ep\log \gamma_i \right\|_{\infty}=\left\|(\tau(\bm U))_i-\ep\log \gamma_i-\lambda_\ep((\tau(\bm U))_i)\right\|_{\infty}\leq 2\|\H\|_{\infty}
	\end{gather*}
	for every $ i \in [N-1]$.
	Moreover, $(\tau(\bm U))_N=\fT_\ep^N(\widehat{\tau (\bm  U)}_N)$ and hence, applying \eqref{eq:reg_i} from Proposition \ref{prop:reg_Hep-transf}, we arrive at
	\begin{align*}
		\left\|(\ren\tau(\bm U))_N  -\ep\log \gamma_N \right\|_{\infty}
		=\left\|(\tau(\bm U))_N -\ep\log \gamma_N+\sum_{j=1}^{N-1}\lambda_\ep((\tau(\bm U))_j)\right\|_{\infty} \leq \|\H\|_{\infty},
	\end{align*}
	which completes the proof.
\end{proof}

\section{Non-commutative multi-marginal optimal transport}\label{sec:mainresults}
%
In this section we prove Theorem \ref{theo:duality}, our first main result stated in Section \ref{sec:contributions}, exploiting the tools developed in Section \ref{sec:prem}. Again, we fix the setup, which remains in force throughout the whole Section \ref{sec:mainresults} and Section \ref{sec:noncomsin}. Let $N \in \bN$, and for $i\in[N]$ we consider density matrices $\gamma_i \in \DM^{d_i}$. Set $\bm\gamma := ( \gamma_i )_{i \in [N]}$, $\bm d = \prod_{j=1}^N d_i$, and assume that $\ker \gamma_i = \{ 0 \}$ (see Remark \ref{rem:kernels}). We also fix $\H \in \cS^{\bm d}$.
%
%
In this section, we prove the Theorem \ref{theo:duality}.
%
%

We begin by introducing the primal functional, which appears in the minimisation \eqref{eq:primal_statements}.

\begin{defin}[Primal Functional]
	Let $\Gamma\in\DM^{{\bf d}}$ the primal functional is defined by
	\begin{equation}\label{eq:primalfcfull} 
		\primal(\Gamma) = \Tr(\H\Gamma) + \ep S(\Gamma) = \Tr(\H\Gamma) + \ep \Tr(\Gamma \log\Gamma).
	\end{equation} 
\end{defin}

We also recall the definitions of the primal and the dual problem
%
\begin{equation}\label{eq:primalproblemfull}
	\prim(\bm \gamma) = \inf \left\lbrace \primal(\Gamma) \suchthat \Gamma\in\DM^{{\bf d}} \text{ and } \Gamma\mapsto (\gamma_1,\dots,\gamma_N)  \right\rbrace,
\end{equation}
\begin{equation}\label{eq:dualproblemfull}
	\dual(\bm \gamma) = \sup \left\lbrace \dualf(\bm U) \suchthat \bm U \in \bigtimes_{i=1}^N \cS^{d_i} \right\rbrace,
\end{equation}
where the dual functional $\dualf$ is given in Definition \ref{eq:dualfcfull}.

\subsection{Primal and dual functionals: lower bound and structure of the optimizers}
\label{sec:duality}
%
%
We begin with the proof of the lower bound for the primal functional  \eqref{eq:primalfcfull}, in terms of the dual functional \eqref{eq:dualfcfull}.
%
%


\begin{proposition}[Lower bound]	\label{prop:lowerbound_multim} Fix $N \in \bN$ and $\ep > 0$. For all $i\in [N]$, let $\gamma_i \in \DM^{d_i}$ be density matrices, $\H \in \cS^{\bm d}$. Then, for all $\bm U \in \bigtimes_{i=1}^N \cS^{d_i}$ and every $\Gamma \in \DM^{\bm d}$, $\Gamma \mapsto \bm \gamma$ we have that
	\begin{align*}
		\primal(\Gamma) \geq \dualf(\bm U) \, .
	\end{align*}
\end{proposition}
%
%
\begin{proof}
	For any $\bm U \in \bigtimes_{i=1}^N \cS^{d_i}$ and any admissible $\Gamma \in \DM^{\bm d}$, $\Gamma \mapsto \bm \gamma$,  we can write
	\begin{align*}
		\primal(\Gamma)
		&=\primal(\Gamma)+\sum_{j=1}^N \Tr(U_j \gamma_j)-\Tr\left(\left(\bigoplus_{j=1}^N U_j\right) \Gamma\right)\\
		&=\sum_{j=1}^N \Tr(U_j \gamma_j)+\ep S(\Gamma)- \Tr\left(\Gamma\left(\bigoplus_{j=1}^N U_j-\H\right)\right).
	\end{align*}
	Let us denote the Hilbert-Schmidt scalar product (on $\cM^{\bm d}$) by $\langle \cdot, \cdot \rangle_{HS}$. It follows that
	\begin{align}
		\label{eq:dualeq1}
		\primal(\Gamma)=\sum_{j=1}^N \Tr(U_j \gamma_j)+\ep\left[S(\Gamma)-\langle \Gamma, \overline Y\rangle_{HS}\right]
		\geq \sum_{j=1}^N \Tr(U_j \gamma_j)-\ep S^*(\overline Y), 
	\end{align}
	where $\overline Y=\ep^{-1}\left(\bigoplus_{j=1}^N U_j-\H\right)\in \cS^{\bm d}$ and, for any $Y \in \cS^{\bm d}$
	\begin{align*}
		S^*(Y):=\sup_{\Gamma \in \cS_\geq^{\bm d}} \{\langle Y, \Gamma \rangle_{HS} -S(\Gamma)\} 
	\end{align*}
	denotes the Legendre transform of $S$ on the subspace $\cS_\geq^{\bm d}$. This can be explicitly computed as
	\begin{align}	\label{eq:S^*}
		S^*(Y)=\Tr\left[\exp(Y-1)\right], \quad \forall Y \in \cS^{\bm d}.
	\end{align}
	
	For the sake of completeness, let us explain how to prove \eqref{eq:S^*}. First of all we show that for any $Y\in \cS^{\bm d}$ the supremum appearing in the definition of $S^*(Y)$ is attained at some $\bar{\Gamma}\in \cS_>^{\bm d}$. Indeed, for any $\Gamma\geq 0$ define $\sigma_+$ to be the maximum of its spectrum, then it holds
	\begin{align*}
		\langle Y, \Gamma \rangle_{HS} -S(\Gamma)\leq \bm d ^2\|Y\|_{\infty} \sigma_+ -\sigma_+ \log \sigma_+ - \min_{\mathbb{R}_+} \{x\log x\} (\bm d^2-1)\xrightarrow{\sigma_+ \to \infty} -\infty.
	\end{align*}
	This implies that the super-levels of $\langle y, \Gamma \rangle_{HS} -f(\Gamma)$ are bounded and hence pre-compact and allows us to conclude the existence of a maximizer $\bar{\Gamma}$. Moreover, it is straightforward to show that $\bar{\Gamma}>0$,  otherwise one would have a contradiction by perturbing $\bar\Gamma$ with $\Pi_{\ker \bar\Gamma}$ (the projector onto $\ker \bar \Gamma$). 
	
	Let us derive the optimality conditions for $\bar \Gamma$. Define $\Gamma_s:= \bar \Gamma+s \Gamma'$ with $\Gamma'\in \cS^{\bm d}$ (note that for any $\Gamma' \in \cS^{\bm d}$ for $s$ sufficiently small $\Gamma_s$ is positive since $\bar \Gamma>0$), then the Euler-Lagrange equation for the maximization problem reads 
	\begin{align*}
		0=\frac d {ds} \Big|_{s=0} \left(\langle Y, \Gamma_s \rangle_{HS} -S(\Gamma_s)\right)=\langle Y, \Gamma' \rangle_{HS}-\Tr\left[\Gamma'(\log \bar\Gamma+1)\right].
	\end{align*}
	%
	This yields $\bar \Gamma=\exp(Y-1)$. Substituting in the expression for $S^*$, we arrive at \eqref{eq:S^*}.
	
	Plugging this into \eqref{eq:dualeq1} with $Y= \overline Y$ and recalling the definition of $Y$, we obtain
	\begin{align*}
		\primal(\Gamma)\geq \sum_{j=1}^N \Tr(U_j \gamma_j)-\ep  \Tr\left(\exp\left(\frac{\bigoplus_{j=1}^N U_j-\H-\ep}{\ep}\right)\right).
	\end{align*}
	Changing the variable $U_1$ to $\tilde{U}_1:=U_1+\ep$, we conclude the proof. 
\end{proof}

\begin{rem}[The non-commutative Schr\"{o}dinger problem]
	\label{rem:nonCommSch_multim}
	Suppose that $\bm U \in \bigtimes_{i=1}^N \cS^{d_i}$ is a fixed point for $\tau$, namely $\tau(\bm U) = \bm U$. This can be equivalently recast as
	$
	\fT_i^\ep ( \hat{\bm U}_i ) = U_i$, $\forall i \in [N].
	$
	Then Lemma \ref{lemmaN:Ucep}, \eqref{eqN:implicit_CNtransform_i} imply that the density matrix defined by 
	\begin{align}	
		\label{eq:Gamma0_multim}
		\Gamma := \exp\left(\frac{ \bigoplus_{i=1}^N U_i- \H}{\ep}\right) 
	\end{align}
	has the correct marginals  $\Gamma \mapsto (\gamma_1, \dots, \gamma_N)$ and thus it is admissible for the primal problem. In particular, it has trace $1$ and we have
	\begin{align*}	
		\dualf(U_1, \dots, U_N)=\dualf(\bm U)= \sum_{i=1}^N \Tr(U_i \gamma_i)= \Tr\left(\left(\bigoplus_{i=1}^N U_i \right) \Gamma\right).
	\end{align*}
	On the other hand, directly from formula \eqref{eq:Gamma0_multim}, we compute
	\begin{align*}
		\Gamma \H + \eps \, \Gamma \log \Gamma = \Gamma \H + \Gamma \left(\bigoplus_{i=1}^N U_i - \H\right) =  \Gamma \left(\bigoplus_{i=1}^N U_i\right)
	\end{align*}
	and thus
	\begin{align}	\label{eq:prim_0_multim}
		\primal(\Gamma) = \Tr\left(\left(\bigoplus_{i=1}^N U_i \right) \Gamma\right)= \dualf(U_1, \dots, U_N).
	\end{align}
	
	In light of Proposition \ref{prop:lowerbound_multim}, this shows that if we are able to find a fixed point of $\tau$, then this must be optimal for the dual problem (note that any maximizer is also a fixed point for $\tau$ as discussed in Remark \ref{rem:maximizersfixed}) and the corresponding $\Gamma$ as obtained in \eqref{eq:Gamma0_multim} must be optimal for the primal problem.
\end{rem}

Another consequence of the above observations is that the set of maximizers for the dual problem is invariant under translations. 
%
%
%
\begin{lemma}[Structure of the maximizers]
	\label{lem:characterizionmax}
	Let $\bm U$ and $\bm V$ be two maximizers of $\dualf$, then there exists $\bm \alpha \in \bR^N$ such that $\sum_{i=1}^N \bm \alpha_i=0$ and 
	$
	\bm U=\bm V +\bm \alpha.
	$	
\end{lemma} 

\begin{proof}
	Thanks to Remark \ref{rem:nonCommSch_multim} and using that the primal functional admits an unique minimizer by strict convexity, we find
	\begin{align}	\label{eq:equality}
		\exp \left( \frac{\bigoplus_{i=1}^N (\bm U)_i -\H} {\ep}\right)=\exp \left( \frac{\bigoplus_{i=1}^N (\bm V)_i -\H} {\ep}\right) \;
		\Longrightarrow \;  \bigoplus_{i=1}^N (\bm U)_i = \bigoplus_{i=1}^N (\bm V)_i \ . 
	\end{align}
	Applying the partial traces to the latter equality, we obtain
	\begin{align*}
		(\bm U)_i=(\bm V)_i+\sum_{j=1, j\neq i}^N \Tr (\bm V)_j-\Tr (\bm U)_j=:(\bm V)_i+\bm \alpha_i.
	\end{align*}
	Using \eqref{eq:equality} once again, one sees that
	\begin{align*}
		\sum_{i=1}^N \bm \alpha_i=(N-1) \left(\Tr\left(	\bigoplus_{i=1}^N (\bm U)_i\right)-\Tr\left(	\bigoplus_{i=1}^N (\bm V)_i\right)\right)=0, 
	\end{align*}
	which concludes the proof.
\end{proof}

\subsection{Proof of Theorem \ref{theo:duality}}
%
%

We are finally ready to prove the equivalence between dual and primal problem, and to characterise the optimisers of the two problems. For the sake of clarity, recall that 
\begin{align*}
	\ren(\bm U)_i = \begin{cases}
		U_i - \lambda_{\ep}(U_i) , &\text{ if } i \in \lbrace 1,\dots,N-1\rbrace \\
		\displaystyle U_N + \sum_{j=1}^{N-1}\lambda_{\ep}(U_i), &\text{ if } i = N, \end{cases}
\end{align*}
as in Definition \ref{def:Ren} and $\lambda_{\ep}$ is defined in \eqref{eq:def_lambda} as $\lambda_\ep(A):= \ep \log \Big( \Tr \Big[ \exp \left( \frac{A}{\ep} \right) \Big] \Big)$, for every $A \in \cS^d$, $d \in \N$. 
%
%
%
%
\begin{proof}[Proof of Theorem \ref{theo:duality}]	$(ii)$. Take a maximizing sequence $\bm U_n$ for the dual problem and consider $\tilde {\bm U}_n := \ren \tau (\bm U_n)$, where $\tau= \Sink_N \circ \dots \circ \Sink_1$ is the Sinkhorn operator as introduced in Definition \ref{def:SinkhornOp}. Thanks to Proposition \ref{propN:uniform_bounds}, $\tilde {\bm U}_n$ is again a maximizing sequence that satisfies
	\begin{align*}
		\big\| \tilde {\bm U}_n \big\|_\infty \leq 2 \| \H \|_\infty + \eps \sup_{i \in [N] } \| \log \gamma_i \|_\infty < \infty, \quad \forall n \in \N, 
	\end{align*}
	and it is therefore compact. Pick any $\bm U^\ep \in \bigtimes_{i=1}^N \cS^{d_i}$ limit point of $\tilde {\bm U}_n$. By continuity of the dual functional we infer
	\begin{align*}
		\dual(\bm \gamma) = \lim_{N \to \infty} \dualf(\tilde {\bm U}_n) = \dualf(\bm U^\ep) 
	\end{align*} 
	which shows that $\bm U^\ep$ is a maximizer for $\dual(\bm \gamma)$. The fact that any other maximizer must coincide with $\bm U^\ep$ follows from Lemma \ref{lem:characterizionmax}.
	
	$(i)\&(iii)$
	Proposition \ref{prop:lowerbound_multim} proves one of the inequalities. To show the other inequality, we take any maximizer $\bm U^\ep$ (which exists by the previous proof of $(ii)$). By construction of the Sinkhorn map, $\bm U^\ep$ must be a fixed point of $\tau$. Thanks to Remark \ref{rem:nonCommSch_multim}, we conclude that
	\begin{align*}
		{\Gamma}^{\ep}=\exp\left( \frac {\bigoplus_{i=1}^N   \bm U_i^\ep-\H}{\ep}\right)
	\end{align*}
	satisfies $ \dualf (\bm U^\eps) = \primal(\Gamma^\eps) \geq \prim(\bm \gamma) $. Hence $\Gamma^\ep$ is optimal for $\primal$ and $\prim(\bm \gamma)  = \dual (\bm \gamma)$.
\end{proof}

\subsection{Stability and the functional derivative of $\prim$}

In this last section, we show stability of the Kantorovich potentials with respect to the marginals $\bm \gamma$ and compute the Fr\'echet differential of $\prim(\bm \gamma)$ (or simply the differential in our finite dimensional setting). A similar result was first obtained by Pernal in \cite{Per05} at zero temperature and in \cite{GieRug19} in the positive temperature 1RDMFT case, i.e. considering also the fermionic and bosonic symmetry constraints. In \cite{Per05}, the result follows by a direct computation via chain rule, by taking the partial derivatives with respect to the eigenvalues and eigenvectors of a density matrix $\Gamma$. On the other hand, \cite{GieRug19} uses tools from convex analysis and exploits the regularity of $\primin$.

Our strategy is based on the Kantorovich formulation of \eqref{eq:primal_statements} and follows ideas contained in \cite{DMaGer20}.


\begin{proof}[Proof of Proposition \ref{prop:OTcont}]
	%
	Consider $\bm \gamma^n \xrightarrow[]{n \to \infty} \bm \gamma$ and pick any sequence of Kantorovich potentials $\bm U^{\ep,n}$ for $\prim(\bm \gamma^n)$. By optimality, they must be a fixed point for $\tau$ and hence, thanks to Proposition \ref{propN:uniform_bounds}, $\ren (\bm U^{\ep,n})$ is uniformly bounded. Note that $\ren (\bm U^{\ep,n})$ are also maximizers for $\dual(\bm \gamma^n)$. This implies that any limit point of $\ren (\bm U^{\ep,n})$ must be a maximizer for $\dual(\bm \gamma)$. The continuity of $\prim(\cdot)$ directly follows from this stability property.
	
	Let us prove the differentiability. Fix $\sigma \in \cS^{d_i}$, with $ \Tr(\sigma) =0$, and denote by $\bm \gamma^h$ the pertubation of $\bm \gamma$ with $+h \sigma$ in the $i$th entry. Denote by $\bm U^\ep$ any Kantorovich potential for $\prim(\bm \gamma)$. From duality (Theorem \ref{theo:duality}) we can estimate
	\begin{align}	\label{eq:differentiab1}
		\frac1h \left( \prim( \bm \gamma^h) - \prim(\bm \gamma) \right) \geq \frac1h \left( \sum_{i=1}^N \Tr \big( \bm U_i^\ep \gamma_i^h - \bm U_i^\ep \gamma_i \big) \right) = \Tr(\bm U_i^\ep \sigma)
	\end{align}
	for every $h \in \R$. Reversely, denote by $\bm U^{\ep,h}$ any sequence of Kantorovich potentials for $\prim(\bm \gamma^h)$. Then for every $h >0$ we obtain
	\begin{align}		\label{eq:differentiab2}
		\frac1h \left( \prim( \bm \gamma^h) - \prim(\bm \gamma) \right) \leq \frac1h \left( \sum_{i=1}^N \Tr \big( \bm U_i^{\ep,h} \gamma_i^h - \bm U_i^{\ep,h} \gamma_i \big) \right) = \Tr(\bm U_i^{\ep,h} \sigma).
	\end{align}
	From the first part of the proof, we know that any limit point of $\ren(\bm U_i^{\ep,h})$ is a Kantorovich potential, which up to translation (Lemma \ref{lem:characterizionmax}) must coincide with $\bm U_i^\ep$. Therefore, passing to the limit in \eqref{eq:differentiab1} and \eqref{eq:differentiab2}, we obtain \eqref{eq:FrechDiff}.
\end{proof}

%
%

\section{Non-commutative Sinkhorn algorithm}\label{sec:noncomsin}
In this section we introduce and prove convergences guarantees (Theorem \ref{theo:introsinkhorn}) of the non-commutative version of the Sinkhorn algorithm, allowing us to compute numerically the minimiser \eqref{eq:optimality_maintheorem} of the non-commutative multi-marginal optimal transport problem \eqref{eq:primalfcfull}.

The idea of the Sinkhorn algorithm is to fix the shape of an ansatz
\begin{align*}
	\Gamma^{(k)} = \exp \left(
	\frac{\bigoplus_{i=1}^N U_i^{(k)} -\H }{\ep}
	\right),
\end{align*}
since it is the actual shape of the minimizer in \eqref{eq:optimality_maintheorem}, and alternately project the Kantorovich potentials $U_i^{(k)}$ via the $(\H,\ep)$-transforms (Definition \ref{def:vecHep}) to approximately reach the constraints $\Gamma^{(k)}\mapsto(\gamma_1,\dots,\gamma_N)$.
Recall that  for $i \in [N]$, the one-step operators $\Sink_i : \bigtimes_{i=1}^N \cS^{d_j} \to \bigtimes_{i=1}^N \cS^{d_j}$ are given by
\begin{align*}
	\bm U:= (U_1, \dots, U_N), \quad \big(\Sink_i(\bm U) \big)_j =
	\begin{cases}
		U_j &\text{if } j \neq i, \\
		\fT_i^\ep(U_1, \dots, U_{i-1}, U_{i+1} , \dots U_N) &\text{if }j=i
	\end{cases}
\end{align*}
where $\fT_i^\ep$ can be implicitly defined (Lemma \ref{lemmaN:Ucep}) solving the equation
\begin{align} \label{eq:def_implicit-T_i}
	\tP_i \left[ \exp\left(\frac {\bigoplus_{i=1}^N \big( \Sink_i(\bm U) \big)_j-\H}{\ep}\right) \right] = \gamma.
\end{align}

\vspace{1mm}
\noindent
\emph{Connection with the multi-marginal Sinkhorn algorithm:} let us shortly describe what is the corresponding picture in the commutative setting \cite{DMaGer19, DMaGer20}. For every $i\in[N]$, let $X_i$ be Polish Spaces, $\rho_i\mathfrak{m}_i \in\Pro(X_i)$ be probability measures with reference measures $\mathfrak{m}_i$. The Hamiltonian $\H$ corresponds to a bounded cost function  $c:X_1\times\dots\times X_N\to\R$.

The Sinkhorn iterates define recursively the sequences $(a^n_j)_{n\in\N}, j\in [N]$ by
\begin{equation}\label{eq:IPFPsequenceN}
	\begin{array}{lcl}
		\ds a^0_j(x_j) & = & \rho_j(x_j), \quad j \in \lbrace 2,\dots,N \rbrace, \\
		\ds a^n_j(x_j) & = & \dfrac{\rho_j(x_j)}{\int \otimes^N_{i<j}a_i^n(x_i)\otimes^N_{i> j}a_i^{n-1}(x_i)e^{-c(x_1,\dots,x_N)/\ep}{\rm d}(\otimes^N_{i\neq j}\mathfrak{m}_i)}, \, \forall n\in \N \text{ and } j \in[N].
	\end{array}
\end{equation}
Via the new variables $u^n_j = \ep\ln(a^n_j), \, j\in[N]$, one can rewrite the Sinkhorn sequences \eqref{eq:IPFPsequenceN} as
\begin{align*}
	u^n_j(x_j) &= - \ep\log\left(\int_{\Pi_{i \neq j}X_i} \exp\left(\frac{\sum_{i\neq j}u^n_i(x_i)-c(x_1,\dots,x_N)}{\ep}\right){\rm d}\left(\otimes^N_{i\neq j}\mathfrak{m}_i\right)\right) + \ep\log(\rho_j) \\
	&= (\hat{u^n_j})^{(N,c,\ep)}(x_j).
\end{align*}
Or, more generally, for every $j\in [N]$, $u^n_j(x_j)$ corresponds to the solution of the maximisation
\[
\argmmax_{u_i\in L^{\infty}(X_i)}\left\lbrace \sum^N_{i=1}\int_{X_j}u_i \rho_j{\rm d}\mathfrak{m}_j -\ep\int_{\Pi^{N}_{i\neq j}X_i}\exp\left(\frac{\sum_{i\neq j}u^n_i+u-c}{\ep}\right){\rm d}\left(\otimes^N_{i\neq j}\mathfrak{m}_i \right) \right\rbrace + \ep\log(\rho_j)
\]
which corresponds to the commutative counterpart of the $i$-th vectorial $(\H,\ep)$-transform in Definition \ref{def:vecHep}.
%

\subsection{Definition of the algorithm}
The non-commutative Sinkhorn algorithm is then defined iterating the $(\H,\ep)$-transforms as in \eqref{eq:def_implicit-T_i} for every $i\in[N]$. 
Note that, by construction, the matrix $\exp\left(\bigoplus_{i=1}^N \left(\left(\Sink_i(\bm U) \right)_j-\H\right)/\ep\right) \in \DM^{\bm d}$ and its $i$-th marginal coincide with $\gamma_i$. 
We define the one-step Sinkhorn map as 
\begin{align*}
	\begin{gathered}
		\tau: \bigtimes_{j=1}^N \cS^{d_j} \to \bigtimes_{j=1}^N \cS^{d_j} , \\
		\tau(\bm U) := (\Sink_N \circ \dots \circ \Sink_1) (\bm U).
	\end{gathered}
\end{align*}

Note that this is the non-commutative counterpart of the iteration defined in \eqref{eq:IPFPsequenceN}.
The Sinkhorn algorithm is obtained iterating the map $\tau$ in the following way.

\vspace{2mm}
\noindent
\textit{Step $0$}. \ We fix $\bm U^{(0)} \in \times_{i=1}^N \cS^{d_i}$ an initial vector of potentials and define the density matrix
\begin{align*}
	\Gamma^{(0)}:= \exp \left( \frac{\bigoplus_{i=1}^N \bm U^{(0)}_i -\H} {\ep}\right) \in \DM^{\bm d}.
\end{align*}

\vspace{1mm}
\noindent
\textit{Step $k$}. 
For every $k \in \N$, we define the $k$-th density matrix via the formula
\begin{align}		\label{eq:kstep_sink}
	\Gamma^{(k)}:= \exp \left( \frac{\bigoplus_{i=1}^N \tau^k(\bm U^{(0)})_i -\H} {\ep}\right) \in \DM^{\bm d} \, ,
\end{align}
where we write $\tau^k := \tau \circ \dots \circ \tau$ the composition of $\tau$ for $k$-times.

\vspace{2mm}
Our goal is to prove the convergence  $\Gamma^{(k)} \to \Gamma^\ep$ where $\Gamma^\ep$ is optimal for $\primin(\bm \gamma)$. To do so, our plan is to obtain compactness at the level of the corresponding dual potentials. Nonetheless, the vectors $\tau^k(\bm U^{(0)})$ do not enjoy good a priori estimates and a renormalisation procedure is needed. For any given sequence $(\bm \alpha^k)_{k \in \N} \subset \R^N$ such that $\sum_{i=1}^N \bm \alpha^k_i =0$, we define
\begin{align}	\label{eq:potentials_alphak}
	\bm U^{(k)}:= \tau^k( \bm U^{(0)}) + \bm \alpha^k, \quad k \in \N,
\end{align}
and observe that, by the properties of $\bigoplus$, the correspond density matrix does not change, thus
\begin{align}	\label{eq:Gammak_equivalence}
	\Gamma^{(k)}= \exp \left( \frac{\bigoplus_{i=1}^N \bm U^{(k)}_i -\H} {\ep}\right) \in \DM^{\bm d}, \quad \forall k \in \N.
\end{align}

Thanks to the good property of the renormalisation map and the Sinkhorn operator, we claim we can find a sequence $\bm \alpha^k$ such that the corresponding potentials $\bm U^{(k)}$ as defined in \eqref{eq:potentials_alphak} do enjoy good a priori estimates and they can be used to prove the convergence of the algorithm, as we see in the next section.

%
\subsection{Convergence guarantees: proof of Theorem \ref{theo:introsinkhorn}}
%

We are ready to prove our main result Theorem \ref{theo:introsinkhorn}, which follows from the next Proposition.
\begin{prop}[Convergence of non-commutative Sinkhorn algorithm]\label{prop:convsink}
	Fix $N \in \bN$ and $\ep > 0$. For all $i\in [N]$, let $\gamma_i \in \DM^{d_i}$ be density matrices, $\H \in \cS^{\bm d}$, with $\ker \gamma_i=\{ 0 \}$. For any initial potential $\bm U^{(0)} \in \bigtimes_{i=1}^N \cS^{d}$, we consider the sequence $\Gamma^{(k)} \in \DM^{\bm d}$ as defined in \eqref{eq:kstep_sink}.
	\begin{enumerate}
		\item There exist $\bm \alpha^k \in \bR^{N}$ with $\sum_{i=1}^N \alpha_i^k =0$ such that
		\begin{align}
			\bm U^{(k)} =\tau^k( \bm U) + \bm \alpha^k\rightarrow \bm U^\ep \quad \text{as } k \to +\infty.
		\end{align}
		\item $\bm U^\ep = (\bm U_1^\ep, \dots, \bm U_N^\ep)$ is optimal for the dual problem $\dual(\gamma)$, as defined in \eqref{eq:dualproblemfull}.
		\item $\Gamma^{(k)}$ converges as $k \to \infty$ to some $\Gamma^\ep \in \DM^{\bm d}$ which is optimal for the primal problem $\prim(\bm \gamma)$, as defined in \eqref{eq:primalproblemfull}. In particular, it holds
		\begin{align}	\label{eq:optim_cond_sect5}
			\Gamma^\ep=\exp\left(\frac {\bigoplus_{i=1}^N \bm U^\ep_i-\H}{\ep}\right) .
		\end{align}
		
	\end{enumerate}

\end{prop}

\begin{proof}
	For any $\bm U^{(0)}\in\bigtimes_{i=1}^N \cS^{d_i}$, we define the sequence
	$
	\bm U_k:= \ren \tau^k(\bm U^{(0)}).
	$ 
	Note that $\bm U_k$ is of the form \eqref{eq:potentials_alphak}, for some $\bm \alpha^k$. Thanks to Proposition \ref{propN:uniform_bounds}, we infer that $\bm U_k$ is uniformly bounded and hence compact. Therefore, there exists a subsequence $\bm U _{k_j} \to \bm U^\ep$. We first show that $\bm U^\ep$ is a maximizer for the dual problem. Indeed, using the properties of $\ren$ and $\tau$, it holds
	\begin{align*}
		\dualf(\tau(\bm U_{k_j}))=\dualf(\tau^{k_j+1}(\bm U^{(0)}))\leq\dualf(\tau^{k_{j+1}}(\bm U^{(0)}))=\dualf(\bm U _{k_{j+1}}).
	\end{align*}
	Passing to the limit the previous inequality, using the continuity of $\dualf$ and $\tau$ and recalling that for any $\bm U$ we have $\dualf(\tau(\bm U))\geq \dualf(\bm U)$, we obtain
	\begin{align*}
		\dualf(\tau( \bm U^\ep ))=\dualf(\bm U^\ep).
	\end{align*}
	By definition, this means that $\bm U^\ep$ is a fixed point for $\tau$ and therefore a maximizer (Remark \ref{rem:nonCommSch_multim}).
	
	In order to prove $(1)$, we show there exists a choice $\bm \alpha^k$ such that $\bm U_k+\bm \alpha^k \to \bm U^\ep$. For $k=k_j$ for some $j$, we pick $\bm \alpha^k=0$, for all the others $k$, we instead pick $\bm \alpha^k$ defined by
	\begin{align*}
		\bm \alpha^k = \argmin_{\bm \alpha} \left\{ \|\bm U_k+\bm \alpha-\bm U^\ep\|_{\infty} \suchthat \sum_{i=1}^N \bm \alpha_i = 0 \right\}.
	\end{align*}
	Note that, by Lemma \ref{lem:characterizionmax}, this is equivalent to picking $\bm \alpha^k$ such that $\bm U^\ep$ is the closest maximizer to $\bm U_k+\bm \alpha^k$.
	We claim this is the right choice. Suppose indeed by contradiction that there exists a subsequence $\bm U_{k'_j}$ such that $\|\bm U_{k'_j}+\bm \alpha_{k_j}-\bm U^\ep\|_{\infty}\geq\delta>0$, then by construction $\|\bm U _{k'_j}+\bm \alpha_{k_j}- \bm U '\|_{\infty}\geq\delta$ for any other maximizer $ \bm U '$. By compactness, this is a contradiction, since there exists a further subsequence $ \bm U _{k''_j}$ of $ \bm U _{k'_j}$ converging to a maximizer $ \bm U '$ (by the same reasoning carried out above). This proves $(1)$ and by optimality of $\bm U^\ep$, $(2)$ as well. The convergence of $\Gamma^{(k)}$ follows from the compactness of $\bm U^{(k)}$ and \eqref{eq:Gammak_equivalence}, whereas the optimality of the limit point $\Gamma^\ep$ and \eqref{eq:optim_cond_sect5} are consequence of the optimality of $\bm U^\ep$ and Remark \ref{rem:nonCommSch_multim}.
\end{proof}
%
%
%

\section{One-body reduced density matrix functional theory}
\label{sec:1RDM}
In this last section, we prove Proposition \ref{prop:pauli} and consequently Theorem \ref{theo:dualBosFerm}.

For given $d,N \in \bN$, we set $\bm d=d^N$ and consider the space of bosonic (resp. fermionic) density matrices $\DM_+^d$ (resp. $\DM_-^d$) as introduced in \eqref{eq:bosonicfermionicDM}. Recall as well that for any given operator $A \in \cS^{\bm d}$, we denote by $A_\pm$ the corresponding projection onto the symmetric space, obtained as $A_\pm := \Pi_\pm \circ A \circ \Pi_\pm$, where $\Pi_\pm$ are defined in \eqref{eq:projectionFermionicBosonic}.

The universal functional in the bosonic and in the fermionic case is then given as in Definition \ref{def:bosons-fermions-primal}, which we recall here for simplicity is given by
\begin{align*}	
	\prim_\pm(\gamma) 
	:= \inf \left\lbrace \Tr(\H\Gamma)+ \ep  \Tr(\Gamma \log \Gamma) \suchthat \Gamma\in\DM_\pm^{\bm d} \text{ and } \Gamma\mapsto \gamma  \right\rbrace \, ,
\end{align*}
whereas the corresponding dual functional and problem (see Definition \ref{def:bosons-fermions-dual}) are given by
\begin{gather*}
	\fbfD(U) 
	:= \Tr( U \gamma ) - \ep
	\Tr \left( 
	\exp \left[ \frac1\ep \bigg(\frac1N  \bigotimes_{i=1}^N U - \H \bigg)_\pm  \right]
	\right) + \ep \, ,\\
	\bfD(\gamma) := \sup \left\{ \fbfD(U) \suchthat U \in \cS^d \right\} \, .
\end{gather*}

We are interested in fully characterizing the existence of the optimizers in the primal and the dual problems, for both bosonic and fermionic cases. Proceeding in a similar way as in the proof of Lemma \ref{lemma:Ucep}, one can prove that every maximizer $U_\pm^\eps$ of the dual functional $\fbfD(\cdot)$ must satisfy the corresponding Euler-Lagrange equation given by
\begin{align}	\label{eq:Euler-lagrange-symmetries}
	\gamma = \tP_1 \left( 
	\exp \left[ \frac1\ep \bigg( \frac1N\bigoplus_{i=1}^N U_\pm^\ep - \H \bigg)_\pm  \right]
	\right)\, .
\end{align}




\subsection{Fermionic dual problem and Pauli's exclusion principle}
\label{sec:dual Pauli}
The aim of this section is to prove Proposition \ref{prop:pauli}. For simplicity we assume, with no loss of generality, that $\eps =1$ and set $\ffd:=\tD_\gamma^{-,1}$. 

For any $U \in \cS^d$, we fix a basis of normalized eigenvectors of $U$, denoted by $\{ \psi_j \}_j$, and consider the decomposition
\begin{align}	\label{eq:eigen}
	U = \sum_{j=1}^d u_j | \psi_j \rangle \langle \psi_j|, \quad u_j \in \R \quad  (\text{eigenvalues}) \, .
\end{align}
We also denote by $\gamma_j:= \langle \psi_j | \gamma | \psi_j \rangle $. In particular, the linear terms read
\begin{align*}
	\Tr (U \gamma) = \sum_{j=1}^d \gamma_j u_j \, .
\end{align*}
For any such basis $\{ \psi_i\}_i$, we obtain a basis of the fermionic tensor product
\begin{gather*}
	\psi^{\emph{as}}_{\bm j}:= \bigwedge_{i=1}^N \psi_{j_i} \, , \quad \bm j = (j_i)_{i=1}^N \in \Theta_- \, , \\
	\Theta_- := \big\{ (j_1, \dots, j_N) \suchthat j_i \in \{ 1, \dots , d \}, \; j_i \neq j_k, \, \text{if } i \neq k \big\} / \S_N \, ,
\end{gather*}
where $\S_N$ denotes the set of permutations of $N$ elements. With respect to this basis, we can write
\begin{align}
	\frac1N \left( \bigoplus_{i=1}^N U \right)_- = \sum_{\bm j \in \Theta_-} \bigg( \frac1N \sum_{i=1}^N u_{j_i} \bigg) | \psi_{\bm j}^{\emph{as}} \rangle \langle \psi_{\bm j}^{\emph{as}} | \, . 
\end{align}
Using the monotonicity of the exponential and the trace, we obtain the following result.

\begin{lemma}[Bounds for $\ffd(U)$]
	Fix $U \in \cS^d$ with eigenvalues $u_j$ and eigenvectors $\{\psi_j\}_j$. For $\gamma \in \DM(d)$, set $\gamma_j:= \langle \psi_j | \gamma | \psi_j \rangle $. Then one has
	\begin{align}	\label{eq:fundam_est}
		\begin{aligned}
			\sum_{j=1}^d \gamma_j u_j - C \sum_{\bm j \in \Theta_-} \exp \left( \frac1N \sum_{i=1}^N u_{j_i} \right)
			&\leq \emph{D}_\gamma^-(U) - 1 \\ &\leq         
			\sum_{j=1}^d \gamma_j u_j - \frac1C \sum_{\bm j \in \Theta_-} \exp \left( \frac1N \sum_{i=1}^N u_{j_i} \right)  ,
		\end{aligned}
	\end{align}
	where $C = \exp \big( \| \H \|_\infty \big) \in (0, +\infty)$.
\end{lemma}
%
Before moving to the proof of Proposition \ref{prop:pauli}, we need the following technical lemma.
\begin{lemma}[Linear term estimates]	\label{lem:linear_est}
	Consider $\{ u_j \}_{j=1}^d \subset \R$ and $\{ \gamma_j\}_{j=1}^d $ such that
	\begin{align}
		\gamma_j \in \Big( \delta, \frac1N - \delta \Big), \quad \sum_{j=1}^d \gamma_j = 1 \, , 
	\end{align}
	for some $\displaystyle \delta \in\Big[0, \frac1{2N}\Big)$. Suppose that $u_j \leq u_k$ if $j \leq k$. Then we have
	\begin{align}	\label{eq:linear_est}
		\sum_{j=1}^d \gamma_j u_j \leq \frac1N \sum_{i=1}^N u_j - \delta (u_1 - u_d) \, .
	\end{align}
\end{lemma}

\begin{proof}
	Thanks to the fact the $u_j$ are ordered, we have the inequality 
	\begin{align*}
		\sum_{j=1}^d \bar \gamma_j u _j \leq \frac1N \sum_{i=1}^N u_j, \quad \forall \;  0 \leq  \bar \gamma_j \leq \frac1N  ,  \quad \sum_{j=1}^d \bar \gamma_j = 1 \, .
	\end{align*}
	Then \eqref{eq:linear_est} follows applying the above inequality to
	\begin{align*}
		\bar \gamma_1 := \gamma_1 + \delta \in \Big( 0,\frac1N \Big), \quad \bar \gamma_j := \gamma_j, \quad \bar \gamma_d := \gamma_d - \delta \in  \Big( 0,\frac1N \Big) \, , 
	\end{align*}
	for every $j \in \{ 2, \dots , d-1\}$.
\end{proof}

We are ready to prove Proposition \ref{prop:pauli}.

\begin{proof}
	($\gamma\leq 1/N\Rightarrow$  $\sup \ffd < \infty$). This is consequence of Proposition \ref{lem:linear_est} with $\delta=0$. More precisely, pick $U \in \cS^d$ and consider a decomposition in eigenfunctions as in \eqref{eq:eigen}. Assume that $\{u_j\}_j$ are non increasing in $j$ (with no loss of generality). We can then apply Proposition \ref{lem:linear_est} with $\delta=0$ and from \eqref{eq:linear_est} and \eqref{eq:fundam_est} we deduce
	\begin{align*}
		\ffd(U) - 1 \leq \frac1N \sum_{i=1}^N u_j - \frac1C \sum_{\bm j \in \Theta_-} \exp \left( \frac1N \sum_{i=1}^N u_{j_i} \right) \leq \frac1N \sum_{i=1}^N u_j - \frac1C \exp \left( \frac1N \sum_{i=1}^N u_j \right) ,
	\end{align*}
	where in the last inequality we used the positivity of the exponential. Therefore
	\begin{align*}
		\sup_{U \in \cS^d} \ffd(U) \leq \sup_{x \in \R} \,  ( x - \frac1C e^x ) + 1 =  \log C < \infty \, .
	\end{align*}
	
	\medskip \noindent
	($\sup \ffd < \infty$ $\Rightarrow \gamma\leq1/N$). Suppose by contradiction that the Pauli's principle is not satisfied. With no loss of generality, we can assume that 
	\begin{align*}
		\gamma = \sum_{i=1}^d \gamma_i |\psi_i \rangle \langle \psi_i|, \quad \gamma_1 > \frac1N \, .
	\end{align*}
	We build the sequence of bounded operators $U^n \in \cS^d$ given by 
	\begin{align}	\label{eq:defUn}
		U^n := \sum_{i=1}^d u_i^n |\psi_i \rangle \langle \psi_i|, \quad u_1^n := n \, , \quad u_j^n := -\frac{n}{N-1} , \quad  \forall j \geq 2 \, .
	\end{align}
	Observe that by construction, we can estimate the non-linear part of $\ffd(U)$ as
	\begin{align*}
		\forall \bm j \in \Theta_-, \quad \exp \left( \frac1N \sum_{i=1}^N u_{j_i} \right) 
		\begin{cases}
			= 1 & \text{if } j_i = 1 \text{ for some }i \,  ,\\
			\leq 1 & \text{otherwise} \, .
		\end{cases}
	\end{align*}
	It follows that we can bound from below $\ffd(U^n)$ as
	\begin{align}	\label{eq:lb}
		\ffd(U^n) \geq  \sum_{j=1}^d \gamma_j u_j^n - C \binom{d}{N} \, .
	\end{align}
	We claim that the linear contribution goes to $+\infty$ as $n \to +\infty$. To see that, note that
	\begin{align}	\label{eq:linear_example}
		\sum_{j=1}^d \gamma_j u_j^n = n \bigg(  \gamma_1 - \frac{1}{N-1} \sum_{i=2}^d \gamma_2 \bigg) = \frac{n}{N-1} \Big(N \gamma_1 - 1 \Big) , 
	\end{align}
	where we used that $\sum_i \gamma_i = 1$. From this, using $\gamma_1 > \frac1N$ and \eqref{eq:lb} we deduce $\ffd(U^n) \to +\infty$ as $n \to +\infty$, thus a contradiction.  \\
	
	\noindent
	(\textit{Equation for the maximizer and uniqueness}). If a maximizer exists, then it solves the equation \eqref{eq:Euler-lagrange-symmetries}. Thanks to the Peierls inequality, 
	we also know that $\ffd$ is strictly concave (because the exponential is strictly convex), hence the uniqueness of the maximizer. \\
	
	\noindent
	(\textit{Existence of} $\Argmax \ffd$ $\Rightarrow 0<\gamma<1/N$ ). We proceed as in the latter proof. By contradiction, assume that 
	\begin{align*}
		\gamma = \sum_{j=1}^d \gamma_j | \psi_j \rangle\langle \psi_j|, \quad \gamma_1 = \frac1N \, , \quad \gamma_j \in \Big( 0 , \frac1N \Big) \, , \quad \forall j \geq 2 \, .
	\end{align*}
	The case $\gamma_j = 0$ can be directly ruled out from the Euler-Lagrange equation for a maximizer \eqref{eq:Euler-lagrange-symmetries}. We can then consider the very same sequence $U^n$ as defined in \eqref{eq:defUn}. From \eqref{eq:lb}, \eqref{eq:linear_example}, and the first part of Theorem \ref{prop:pauli}, on one hand we deduce
	\begin{align*}
		- C  \binom{d}{N} \leq  \ffd(U^n) \leq \log C \, , \quad \forall n \in \mathbb N \, .
	\end{align*}
	On the other hand, $\| U^n \|_\infty \to +\infty$ as $n \to \infty$, which means that $\ffd$ is not coercive. Thanks to Peierls inequality, we also know that $\ffd$ is strictly concave, which implies that $\ffd$ can not  attain its maximum. \\
	
	\noindent
	($0<\gamma<1/N\Rightarrow$ \textit{existence of} $\Argmax \ffd$). Let $U \in \cS^d$ and consider a decomposition in eigenfunctions as in \eqref{eq:eigen}. Assume that $\{u_j\}_j$ are non increasing in $j$ (with no loss of generality) and denote by $\gamma_j:= \langle \psi_j | \gamma | \psi_j \rangle $. By assumption, there exists $\delta \in \big( 0 , \frac1N \big)$ such that
	\begin{align}
		\sum_{j=1}^d \gamma_j = 1 , \quad \gamma_j \in \Big( \delta, \frac1N - \delta \Big) \, , \quad \forall j \in \{1, \dots, d\} \, .
	\end{align}
	
	We can then apply Proposition \ref{lem:linear_est} and \eqref{eq:fundam_est} to obtain
	\begin{align}	\label{eq:sharp_ub}
		\ffd(U) -1 
		&\leq \frac1N \sum_{i=1}^N u_j - \frac1C \sum_{\bm j \in \Theta_-} \exp \left( \frac1N \sum_{i=1}^N u_{j_i} \right) - \delta (u_1 - u_d)  \\ 
		&\leq \frac1N \sum_{i=1}^N u_j - \frac1C \exp \left( \frac1N \sum_{i=1}^N u_j \right)  - \delta (u_1 - u_d)  ,
	\end{align}
	where we used the positivity of the exponential. Set $S:= \sup_x (x - \frac{e^x}{C}) + 1<\infty$, and infer 
	\begin{align}	\label{eq:coercivity_bound}
		\ffd(U)  \leq S - \delta (u_{max} - u_{min}), \quad \forall U \in \cS^d, \quad  U = \sum_{j=1}^d u_j | \psi_j \rangle \langle \psi_j|  \, ,
	\end{align}
	where $u_{max}$ and $u_{min}$ denotes respectively the maximum/minimum eigenvalue of $U$. Let us use this estimate to prove to existence of a maximizer for $\ffd$. Consider a maximizing sequence $U^n$ of bounded operators. In particular, we can assume that $-I:= \inf_n \ffd(U^n) \geq -\infty$. If the sequence $\{ U^n \}_n$ is bounded in $\cS^d$, then any limit point is a maximum for $\ffd$, by concavity and continuity of $\ffd$, and the proof is complete. Suppose by contradiction that $\| U^n \|_\infty \to +\infty$ as $n \to +\infty$. Note that from \eqref{eq:coercivity_bound} we deduce
	\begin{align}
		\sup_{n \in \mathbb N} \big( u_{max}^n - u_{min}^n \big) \leq \frac{S + I}{\delta} < \infty  \, ,
	\end{align}  
	therefore we deduce that either $u_j^n \to -\infty$ or $u_j^n \to +\infty$ for every $j \in \{ 1, \dots, d \}$.
	In the first case, we would have a contradiction, because
	\begin{align*}
		- I \leq \ffd(U^n) \leq \sum_{j=1}^d \gamma_j u_j^n + 1 \to -\infty \quad \text{as }n \to +\infty \, .
	\end{align*}
	In the second case, we can use \eqref{eq:sharp_ub} to find a contradiction, because
	\begin{align*}
		-I \leq \ffd(U^n) \leq \frac1N \sum_{i=1}^N u_j^n - \frac1C \exp \left( \frac1N \sum_{i=1}^N u_j^n \right) +1 \to -\infty \, , 
	\end{align*}
	where we used that $\displaystyle \lim_{x \to +\infty} (x - C^{-1} e^x) = -\infty$. The proof is complete.
\end{proof}

\subsection{Duality theorem for fermionic and bosonic systems}
In this section we prove Theorem \ref{theo:dualBosFerm}. The proof relies on the use of Theorem \ref{theo:duality} and the existence of maximizers for $\ffD$, proved in Proposition \ref{prop:pauli}, and $\fbD$. The latter can be proven easily by noting that the spectrum of $\left(\bigoplus_{j=1}^N U_j\right)_+$ contains the spectrum of $U$ and, applying similar computations to the ones used in the case of $\ffD$, deducing the coercivity of $\fbD$. We also need the following observation. 

\begin{rem}
	If $\H$ satisfies \eqref{eq:symmetry} and $\bm \gamma = (\gamma_i)_i$ , $
	\gamma_i=\gamma$, then the minimizers of $\dualf$ (the dual functional without symmetry constraints) can be taken to satisfy $U_i \equiv U$, for some $U \in \cS^d$. In particular
	\begin{align*}
		\dual (\bm \gamma) 
		= \sup_{\bm U \in (\cS^{\bm d})^N} \dualf(\bm U) 
		= \sup_{U \in \cS^d} \left\{ 
		\Tr( U \gamma ) - \ep
		\Tr \left( 
		\exp \bigg[ \frac1\ep\bigg( \frac1N \bigotimes_{i=1}^N U - \H \bigg) \bigg]
		\right) 
		\right\}  + \ep\, .
	\end{align*}
	This follows from the observation that if $\bm U \in (\cS^{\bm d})^N$, then we obtain a symmetric competitor $\tilde {\bm U}$
	\begin{align*}
		(\tilde {\bm U})_i = \frac1N \sum_{j=1}^N U_j \, ,  \quad \text{such that} \quad \dualf(\tilde {\bm U}) = \dualf(\bm U) \, .
	\end{align*}
\end{rem}

\begin{proof}[Proof of Theorem \ref{theo:dualBosFerm}]
	Let us assume that $\gamma >0$ in the bosonic case ($0 < \gamma <\frac1N$ in the fermionic case). The general duality result (including the case $\gamma$ in which does not satisfy the above strict inequalities) can be handled by decomposition of the space, in the same way as in  Remark \ref{rem:kernels}. 
	
	Under these assumptions, thanks to Proposition \ref{prop:pauli}, we know that a maximizer $U_\pm^\ep$ exists and satisfies \eqref{eq:Euler-lagrange-symmetries}.
	
	We then define the $N$-particle density matrix 
	\begin{align*}
		\tilde \Gamma_\pm^\ep := \exp \bigg[ \frac1\ep \bigg( \frac1N\bigoplus_{i=1}^N U_\pm^\ep - \H \bigg)_\pm \bigg]  \in \cS^{\bm d}\, ,
	\end{align*}
	and thanks to Remark \ref{rem:nonCommSch_multim}, we know that $\tilde \Gamma_\pm^\ep$ is optimal for the problem $\prim\big(\tP_1 ( \tilde \Gamma_\pm^\ep) \big)$ without symmetry contraints. Observing that $(\tilde \Gamma_\pm^\ep)_\pm = \Gamma_\pm^\ep$ (defined in \eqref{eq:optimum_primal_fermions},\eqref{eq:optimum_primal_bosons}), we deduce that $\Gamma_\pm^\ep$ must be optimal for the primal problem $\prim_\pm(\gamma)$ with symmetry constraints. This also proves the equality between primal and dual problems and concludes the proof. 
\end{proof}

\newpage
\section*{Acknowledgments}
This work started when A.G. was visiting the Erwin Schr\"odinger Institute and then continued when D.F. and L.P visited the Theoretical Chemistry Department of the Vrije Universiteit Amsterdam. The authors thanks the hospitality of both places and, especially, P. Gori-Giorgi and K. Giesbertz for fruitful discussions and literature suggestions in the early state of the project. Finally, the authors also thanks J. Maas and R. Seiringer for their feedback and useful comments to a first draft of the article.

L.P. acknowledges support by the Austrian Science Fund (FWF), grants No W1245 and No F65.
D.F acknowledges support by the European Research Council (ERC) under the European Union's Horizon 2020 research and innovation programme (grant agreements No 716117 and No 694227). A.G. acknowledges funding by the European Research Council under H2020/MSCA-IF “OTmeetsDFT” [grant ID: 795942].

\end{document}